\documentclass[dvipsnames]{article}
\usepackage{arxiv, times}
\usepackage{graphicx}
\usepackage{balance}  
\usepackage{booktabs} 
\usepackage{blkarray} 
\usepackage{amsmath,amsfonts,amssymb}
\usepackage{amsthm}
\usepackage{mathtools}
\usepackage{algorithm}
\usepackage[noend]{algpseudocode}

\usepackage[utf8]{inputenc}

\graphicspath{{figs/}}

\algblock{ParFor}{EndParFor}
\algnewcommand\algorithmicparfor{\textbf{parallel for}}
\algnewcommand\algorithmicpardo{\textbf{do}}
\algnewcommand\algorithmicendparfor{\textbf{end parfor}}
\algrenewtext{ParFor}[1]{\algorithmicparfor\ #1\ \algorithmicpardo}
\algrenewtext{EndParFor}{\algorithmicendparfor}

\makeatletter
\ifthenelse{\equal{\ALG@noend}{t}}%
  {\algtext*{EndParFor}}
  {}%
\makeatother

\newcommand{\push}[1]{\text{push} \left ( #1 \right )}
\newcommand{\pop}{\text{pop}()}

\newcommand{\algoname}[1]{\textnormal{\textsc{#1}}}

\newcommand{\E}{\mathbb{E}}
\newcommand{\Var}{\mathrm{Var}}

\newcommand{\papertitle}{Degree{S}ketch: Distributed Cardinality Sketches on Massive Graphs with Applications}

\usepackage{hyperref}
\hypersetup{
    unicode=false,
    pdftoolbar=true,
    pdfmenubar=true,
    pdffitwindow=false,
    pdfstartview={FitH},
    pdftitle={\papertitle},
    pdfauthor={Benjamin W. Priest},
    pdfsubject={\papertitle},
    pdfkeywords={Sketching algorithms, Distributed Algorithms, Massive Graphs},
    pdfnewwindow=true,
    colorlinks=True,
    linkcolor=BurntOrange,
    citecolor=RawSienna,
    filecolor=magenta,
    urlcolor=blue
}
\newcommand{\email}[1]{\tt\small\href{mailto:#1}{#1}}

\title{\papertitle}

\author{
Benjamin W. Priest\\[0.3cm]
       {Center for Applied Scientific Computing, Lawrence Livermore National Laboratory}\\[0.2cm]
       \texttt{\email{priest2@llnl.gov}}
}

\date{}

\iclrfinalcopy

\begin{document}

\maketitle

\begin{abstract}
We present \algoname{DegreeSketch}, a semi-streaming distributed sketch datastructure and demonstrate its utility for estimating local neighborhood sizes and local triangle count heavy hitters on massive graphs.
\algoname{DegreeSketch} consists of vertex-centric cardinality sketches distributed across a set of processors that are accumulated in a single pass, and then behaves as a persistent query engine capable of approximately answering graph queries pertaining to the sizes of adjacency set unions and intersections.
The $t$th local neighborhood of a vertex 
is the number of vertices reachable in $G$ from $v$ by traversing at most $t$ edges, whereas the local triangle count 
is the number of 3-cycles in which it is included.
Both metrics are useful in graph analysis applications, but exact computations scale poorly as graph sizes grow. 
We present efficient algorithms for estimating both local neighborhood sizes and local triangle count heavy hitters using \algoname{DegreeSketch}.
In our experiments we implement \algoname{DegreeSketch} using the celebrated hyperloglog cardinality sketch and utilize the distributed communication tool YGM to achieve state-of-the-art performance in distributed memory. 
\end{abstract}

\section{Introduction and Related Work}
 \label{sec:intro}

As graph datasets scales continue to grow in applications, basic queries are becoming increasingly difficult to answer.
How connected are the proteins in an interaction network in aggregate?
Which hyperlinks shortcut large numbers of possible intermediate webpages?
How many friends of friends of friends does a particular profile in a social network have?
Many such queries amount to reasoning about the unions and intersections of the neighbor sets of the vertices in a graph.
However, answering such queries exactly is typically superlinear in compute time and communication in distributed implementations, untenable for massive graphs.

Furthermore, the simple \emph{storage} of large graphs can become burdensome, particularly as scale-free graphs include vertices whose degree is linear in the size of the graph.
Not to mention, communicating neighborhood set information about such vertices is impractical.
It is therefore tempting to consider schemata for sublinearly summarizing the information contained in vertex adjacency sets, and estimating unions and intersections.
It is known that any data structure that provides relative error guarantees for the cardinality of a multiset with $n$ unique elements requires $O(n)$ space \citep{alon1999space}.
Consequently, investigators have developed many so-called \emph{cardinality sketches} that provide such relative error guarantees in $o(n)$ space while admitting a small probability of failure, such as PCSA \citep{flajolet1985probabilistic}, MinCount \citep{bar2002counting},  LogLog \citep{durand2003loglog}, Multiresolution Bitmap \citep{estan2003bitmap}, HyperLogLog \citep{flajolet2007hyperloglog}, and the space-optimal solution of \citep{kane2010optimal}.
While all these cardinality sketches have a natural union operation that allows one to combine the sketches of two multisets into a sketch of their union, most have no closed intersection operation.
Many, however, admit a heuristic intersection estimator that cannot obtain bounded error due to known lower bounds.

We present the \algoname{DegreeSketch} data structure, which maintains a cardinality sketch for each vertex distributed over a set of processors. 
These sketches accumulate in a single pass over a data stream describing the graph, and use a total amount of space polyloglinear in the number vertices - i.e. \algoname{DegreeSketch} is a semi-streaming data structure. 
We demonstrate its utility with distributed algorithms that estimate local $t$-neighborhood sizes, as well as edge- and vertex-local triangle count heavy hitters.

The local $t$-neighborhood of a vertex is the number of vertices that can be reached in $t$ hops.
As $t$ increases, the $t$-neighborhoods describe how the ``ball'' around the vertex grows.
Knowledge of these ball sizes can be useful for applications such as edge prediction in social networks \citep{gupta2013wtf} and probabilistic distance calculations \citep{boldi2011hyperanf, myers2014information}.
For example, knowing the 3-neighborhood size of a user profile in a social network predicts cost of performing a computation over the set of its friends of friends of friends.  
The \algoname{ANF} \citep{palmer2002anf} and \algoname{HyperANF} \citep{boldi2011hyperanf} algorithms estimate the neighborhood function, the ``average ball'' around vertices in a graph, by individually estimating and summing the local $t$-estimates for all vertices.
using Flajolet-Martin and \algoname{HyperLogLog} cardinality sketches, respectively.
We present an algorithm producing an estimate of a similar form, but its distributed implementation allows it to scale to much larger graphs. 
Moreover, \algoname{DegreeSketch} is a leave-behind reusable data structure.

Counting the number of triangles in simple graphs is a canonical problem in network science.
A ``triangle'' is a trio of co-adjacent vertices, and is the smallest nontrivial community structure.
Both the global count of triangles and the vertex-local counts, i.e. the number of triangles incident upon each vertex, are key to network analysis and graph theory topics such as cohesiveness \citep{lim2015mascot}, global and local clustering coefficients \citep{tsourakakis2008fast}, and trusses \citep{cohen2008trusses}.
Local triangle counts are useful in protein interaction analysis \citep{milo2002network},  spam detection \citep{becchetti2010efficient}, and community discovery \citep{wang2010triangulation, berry2011tolerating}.

Although many exact algorithms have been proposed for the triangle counting problem \citep{tsourakakis2008fast, becchetti2010efficient, chu2011triangle, suri2011counting, wolf2017fast}, their time complexity is superlinear in the number of edges ($O(m^{\frac{3}{2}})$).
In order to avoid this dreaded superlinear scaling for applications involving large graphs, many researchers have turned to streaming approximations.
These serial streaming algorithms  maintain a limited number of sampled edges from an edge stream.
Streaming global triangle estimation algorithms have arisen that sample edges with equal probability \citep{tsourakakis2009doulion}, sample edges with probability relative to counted adjacent sampled edges and incident triangles \citep{ahmed2017sampling}, and sample edges along with paths of length two \citep{jha2013space}. 
The first proposed semi-streaming local triangle estimation algorithm relies upon min-wise independent permutations accumulated over a logarthmic number of passes \citep{becchetti2008efficient}. 
More recently, true single-pass sampling algorithms have arisen such as \algoname{Mascot} \citep{lim2015mascot}
and \algoname{Tri\'est} \citep{stefani2017triest}.

While many distributed global and vertex-local triangle counting algorithms have been proposed, the overwhelming majority store the graph in distributed memory and return exact solutions \citep{suri2011counting, arifuzzaman2013patric, pearce2017triangle}.
Recently, the study of distributed streaming vertex-local triangle counting was intiated in earnest with the presentation of \algoname{Try-Fly} \citep{shin2018tri}, a parallelized generalization of \algoname{Tri\'est}, and its follow-up \algoname{DiSLR} that introduced limited edge redundancy \citep{shin2018dislr}.
Our approach is fundamentally different to these methods, depending upon sketching rather than sampling as its core primitive. 
\algoname{DegreeSketch} also permits the estimation of edge-local triangle counts, or the number of triangles in which individual edges participate.
While the sampling approaches produce estimates and a stochastically sampled sparse graph, we produce a leave-behind queryable data structure.

We begin with a discussion of some of the preliminaries and notation in Section~\ref{sec:notation}.
Section~\ref{sec:DS} introduces \algoname{DegreeSketch}, as well as describing algorithms that utilize \algoname{DegreeSketch} to perform local neighborhood size estimation as well as approximately recovering edge- and vertex-local triangle count heavy hitters.  
Section~\ref{sec:HLL} describes the details of a particular implementation of \algoname{DegreeSketch} using \algoname{HyperLogLog} cardinality sketches.
We conclude with experiments in Section~\ref{sec:experiments}.

\section{Preliminaries and Notation}
 \label{sec:notation}

Throughout this document we will consider an undirected graph $\mathcal{G} = (\mathcal{V}, \mathcal{E})$, which we assume to be large.
We adopt the usual convention that  $|\mathcal{V}| = n$ and $|\mathcal{E}| = m$.
Let $d_\mathcal{G}(x, y)$ be the length of the shortest path in $\mathcal{G}$ between $x, y \in \mathcal{V}$, and let $\mathbf{d}(x)$ be the degree of $x$.
We will consider a universe of processors $\mathcal{P}$, and further assume a partitioning of vertices to processors $f : \mathcal{V} \rightarrow \mathcal{P}$. 
We will occasionally abuse notation and use $f^{-1}(P)$ to describe the set of vertices that map to $P \in \mathcal{P}$.
We make no assumptions about the particulars of $f$, noting that vertex partitions are a subject of intense academic scrutiny.
In effect, our algorithms are designed to work alongside any reasonable $f$. 
In our algorithms we will occasionally use the keyword $\algoname{Reduce}$ to refer to a global sum, except were the operand is a max heap in which case it is the creation of a global max heap.

We assume that $\mathcal{G}$ is given by a data stream $\boldsymbol{\sigma}$, a sequential list of the edges of $\mathcal{E}$.
$\boldsymbol{\sigma}$ is further partitioned by some unknown means into $|\mathcal{P}|$ substreams, one to be read by each processor.
We assume that each processor $P \in \mathcal{P}$ has send and receive buffers $\mathcal{S}[P]$ and $\mathcal{R}[P]$, respectively.
Algorithms are given broken up into \textbf{Send}, \textbf{Receive}, and local \textbf{Computation} \textbf{Contexts}.
We make no assumptions as to how processors handle switching between contexts.
In our implementations we use the software package \algoname{YGM} \citep{priest2019you} to manages the send and receive buffers, as well as switching contexts, in a manner that is opaque to the client algorithm.

Consider a vertex $x \in \mathcal{V}$.
For $t \in \mathbf{N}$, let $\mathcal{N}_\mathcal{G}(x, t)$ be the local $t$-neighborhood of $x$ defined as 
\begin{equation} \label{eq:func:nbhd:local}
	\mathcal{N}_\mathcal{G}(x, t) = |\{y \in \mathcal{V} | d_\mathcal{G}(x, y) \leq t\}|.
\end{equation}
Moreover, let the \emph{global} $t$-neighborhood $\mathcal{N}_\mathcal{G}(t)$ be defined by 
\begin{equation} \label{eq:func:nbhd:sum}
	\mathcal{N}_\mathcal{G}(t) = \sum_{x \in \mathcal{V}} \mathcal{N}_\mathcal{G}(x, t).
\end{equation}
We define the edge-local triangle count of an edge $xy \in \mathcal{E}$ as
\begin{equation}
	\mathcal{T}_\mathcal{G}(xy) = |\{z \in \mathcal{V} \setminus \{x,y\} \mid yz, zx \in \mathcal{E}\}|.
\end{equation}
This allows us to define the more con
We define the vertex-local triangle count of $x$ as
\begin{equation}
	\mathcal{T}_\mathcal{G}(x) = |\{yz \in \mathcal{E} \mid xy, zx \in \mathcal{E} \wedge |\{x, y, z\}| = 3\}|.
\end{equation}
We can equivalently define vertex-local triangle counts in terms of edge-local triangle counts as
\begin{equation} \label{eq:sum_of_edge_tris}
	\mathcal{T}_\mathcal{G}(x) 
	= \frac{1}{2}\sum\limits_{xy \in \mathcal{E}} \mathcal{T}_\mathcal{G}(xy).
\end{equation}
We will also refer to the global number of triangles in a graph 
\begin{equation} \label{eq:global_tris}
	\mathcal{T}_\mathcal{G} 
	= \frac{1}{3}\sum_{x \in \mathcal{V}} \mathcal{T}_\mathcal{G}(x)
	= \frac{1}{3}\sum_{xy \in \mathcal{E}} \mathcal{T}_\mathcal{G}(xy)
\end{equation}
We will drop the subscripts where they are clear.

As we have stated, \algoname{DegreeSketch} consists of a set of cardinality sketches.
Many such sketches would suffice to implement DegreeSketch, so for the purposes of discussion we will abstract many of the particulars until Section~\ref{sec:HLL}.
We will assume a notional sketch that requires $O(\varepsilon^{-2})$ space, where $\varepsilon$ is an accuracy parameter.
We assume that this sketch supports an $\algoname{Insert}(\cdot)$ operation to add elements and admits a $\widetilde{| \cdot |}$ operator, providing an approximation that is within a multiplicative factor of $[1-\varepsilon, 1+\varepsilon]$ of the number of unique inserted items with high probability.
We also assume that the sketch affords a closed $\widetilde{\cup}$ operator to combine sketches, and a $\widetilde{| \cdot \cap \cdot |}$ operator to estimate intersection cardinalities.
For reasons that will be described in Section~\ref{sec:HLL:intersection}, we do not assume that the intersection operator has the same error properties as $\widetilde{| \cdot |}$.

\section{DegreeSketch}
 \label{sec:DS}

\algoname{DegreeSketch} maintains a distributed data structure $\mathcal{D}$ that can be queried for an estimate of a vertex's degree.
For each $x \in \mathcal{V}$ we maintain a cardinality sketch $\mathcal{D}[x]$, which affords the approximation of $\mathbf{d}(x)$. 
We will assume throughout that $f(x) \in \mathbf{P}$ is the processor that will hold $\mathcal{D}[x]$.

Algorithm~\ref{alg:ds:accumulation} describes the distributed accumulation of a \algoname{DegreeSketch} instance on a universe of processors $\mathcal{P}$.
In a distributed pass over the partitioned stream, processors use the partition function $f$ to send edges to the cognizant processors for each endpoint. 
These processors each maintain cardinality sketches for their assigned vertices.
When $P \in \mathcal{P}$ receives an edge $xy \in \mathcal{E}$ where $f(x) = P$, it performs $\algoname{Insert}(\mathcal{D}[x], y)$.
Once all processors are done reading and communicating, $\mathcal{D}$ is accumulated.

\begin{algorithm}[t] 
\caption{Accumulation}\label{alg:ds:accumulation}
\begin{flushleft}
        \textbf{Input:} 		$\boldsymbol{\sigma}$ - edge stream divided into $|\mathcal{P}|$ substreams\\
        	\hspace{3.2em}	$\mathcal{P}$ - universe of processors	 \\
        	\hspace{3.2em}	$\mathcal{S}$ - dictionary mapping $\mathcal{P}$ to send queues	 \\
        	\hspace{3.2em}	$\mathcal{R}$ - dictionary mapping $\mathcal{P}$ to receive queues	 \\
        	\hspace{3.2em}	$f$ - partition mapping $\mathcal{V} \rightarrow \mathcal{P}$	 \\
        \textbf{Returns:} $\mathcal{D}$ - accumulated DegreeSketch
\end{flushleft}
\begin{flushleft}
\begin{algorithmic}[1]
	\Statex \textbf{Send Context} for $P \in \mathcal{P}$:
		\While{$\mathcal{S}[P]$ is not empty}
  			\State $(f(x), xy) \gets \mathcal{S}[P].\pop$
			\State $\mathcal{R}[f(x)].\push{xy}$
  		\EndWhile
	\Statex \textbf{Receive Context} for $P \in \mathcal{P}$:
		\While{$\mathcal{R}[P]$ is not empty}
  			\State $xy \gets \mathcal{R}[P].\pop$
  			\If {$!\exists \mathcal{D}[x]$} 
				$\mathcal{D}[x] \gets \text{empty sketch}$
			\EndIf
	  		\State $\algoname{Insert}(\mathcal{D}[x], y)$
  		\EndWhile
	\Statex \textbf{Computation Context} for $P \in \mathcal{P}$:
		\State $\mathcal{D} \gets $ empty $\algoname{DegreeSketch}$ dictionary
		\While{$\sigma_P$ has unread element $uv$}
 			\State $\mathcal{S}[P].\push{f(u), uv}$
 			\State $\mathcal{S}[P].\push{f(v), vu}$
		\EndWhile
		\State \Return $\mathcal{D}$
\end{algorithmic}
\end{flushleft}
\end{algorithm}

\algoname{DegreeSketch} can be implemented with any cardinality sketch that admits some form of close union operator and intersection estimation.
In fact, the algorithms in Sections~\ref{sec:edge_triangles} and ~\ref{sec:vertex_triangles} do not even require a closed union operator.
In our experiments, we focus on the well-known \algoname{HyperLogLog} or \algoname{HLL} cardinality sketches.

We describe \algoname{HLL} and discuss these features in greater detail in Section~\ref{sec:HLL}.
First, however, we describe algorithms utilizing \algoname{DegreeSketch} for neighborhood size estimation in Section~\ref{sec:neighborhoods}, recovering edge-local triangle count heavy hitters in Section~\ref{sec:edge_triangles}, and finally recovering vertex-local triangle count heavy hitters in Section~\ref{sec:vertex_triangles}.

\subsection{Neighborhood Size Estimation}
 \label{sec:neighborhoods}

Let $\mathcal{D}$ be an instance of \algoname{DegreeSketch} as described, so that for $x \in \mathcal{V}$, $\mathcal{D}[x]$ is a cardinality sketch of the adjacency set of $x$.
By the properties of cardinality sketches, if we know $x$'s neighbors we can compute an estimate of $\mathcal{N}(x, 2)$ by computing 

\begin{equation} \label{eq:estimate:nbhd}
	\widetilde{\mathcal{N}}_(x, 2) 
	= \widetilde{\left | \widetilde{\bigcup}_{y: xy \in \mathcal{E}} \mathcal{D}[y] \right |}.
\end{equation}

\begin{algorithm}[htbp] 
\caption{Neighborhood Approximation}\label{alg:ds:anf}
\begin{flushleft}
        \textbf{Input:} 		$\boldsymbol{\sigma}$ - edge stream divided into $|\mathcal{P}|$ substreams\\
        	\hspace{3.2em}	$\mathcal{P}$ - universe of processors	 \\
        	\hspace{3.2em}	$\mathcal{D}^1$ - accumulated \algoname{DegreeSketch} 	 \\
        	\hspace{3.2em}	$\mathcal{S}$ - dictionary mapping $\mathcal{P}$ to send queues	 \\
        	\hspace{3.2em}	$\mathcal{R}$ - dictionary mapping $\mathcal{P}$ to receive queues	 \\
        	\hspace{3.2em}	$f$ - partition mapping $\mathcal{V} \rightarrow \mathcal{P}$	 \\
        	\hspace{3.2em}	$k$ - maximum neighborhood degree	 \\
        \textbf{Returns:} $\widetilde{\mathcal{N}}(x, t)$, $\widetilde{\mathcal{N}}(t)$ for all $x \in \mathcal{V}$, $t \leq k$
\end{flushleft}
\begin{flushleft}
\begin{algorithmic}[1]
	\Statex \textbf{Send Context} for $P \in \mathcal{P}$:
		\While{$\mathcal{S}[P]$ is not empty}
			\If {next message is an \algoname{Edge}}
	  			\State $(f(x), xy, t) \gets \mathcal{S}[P].\pop$
				\State $\mathcal{R}[f(x)].\push{\algoname{Edge}, (xy, t)}$
			\ElsIf{next message is a \algoname{Sketch}}
				\State $(f(y), \mathcal{D}^{t-1}[x], y, t) \gets \mathcal{S}[P].pop()$
				\State $\mathcal{R}[f(y)].\push{\algoname{Sketch}, (\mathcal{D}^{t-1}[x], y, t)}$
			\EndIf
  		\EndWhile
	\Statex \textbf{Receive Context} for $P \in \mathcal{P}$:
		\While{$\mathcal{R}[P]$ is not empty}
			\If {next message is an \algoname{Edge}}
	  			\State $(xy, t) \gets \mathcal{R}[P].\pop$
				\State $\mathcal{S}[P].\push{\algoname{Sketch}, (f(y), \mathcal{D}^{t-1}[x], y, t))}$
			\ElsIf{next message is a \algoname{Sketch}}
				\State $(\mathcal{D}^{t-1}[x], y, t) \gets \mathcal{R}[P].\pop$
				\State $\mathcal{D}^t[y] \gets \mathcal{D}^t[y] \widetilde{\cup} \mathcal{D}^{t-1}[x]$
			\EndIf
  		\EndWhile
	\Statex \textbf{Computation Context} for $P \in \mathcal{P}$:
		\State $t \gets 1$
		\While{true}
			\State $\widetilde{\mathcal{N}}(x, t) \gets \widetilde{|\mathcal{D}^t[x]|}$ for $x \in f^{-1}(P)$
			\State $\widetilde{\mathcal{N}}(t) \gets \sum\limits_{x \in f^{-1}(P)} \widetilde{\mathcal{N}}(x, t)$ \label{alg:ds:anf:line:sum2}
			\State $\widetilde{\mathcal{N}}(t) \gets$ \algoname{Reduce} $\widetilde{\mathcal{N}}(t)$
			\State $t \gets t + 1$
			\If {$t > k$} 
			\textbf{break}
			\EndIf
			\State Reset $\mathcal{\sigma}_P$
			\State $\mathcal{D}^t \gets \mathcal{D}^{t - 1}$
			\While{$\sigma_P$ has unread element $uv$}
	 			\State $\mathcal{S}[P].\push{\algoname{Edge}, (f(u), uv, t)}$
 				\State $\mathcal{S}[P].\push{\algoname{Edge}, (f(v), vu, t)}$
			\EndWhile
		\EndWhile
\end{algorithmic}
\end{flushleft}
\end{algorithm}

Higher-order union operations of the form Equation~(\ref{eq:estimate:nbhd}) are the core of the \algoname{ANF} \citep{palmer2002anf} and \algoname{HyperANF} \citep{boldi2011hyperanf} algorithms.
Algorithm~\ref{alg:ds:anf} is a distributed generalization of \algoname{HyperANF} using \algoname{DegreeSketch}.
After accumulating $\mathcal{D}^1$, an instance of \algoname{DegreeSketch}, the algorithm takes a number of additional passes over $\boldsymbol{\sigma}$. 
For $t$ starting at 2, we accumulate 
\begin{equation} \label{eq:nextlayer}
\mathcal{D}^t[x] = \widetilde{\bigcup}_{y: xy \in \mathcal{E}} \mathcal{D}^{t-1}[y]
\end{equation}
by way of a message-passing scheme similar to Algorithm~\ref{alg:ds:accumulation}.
When $P \in \mathcal{P}$ receives an edge $xy \in \mathcal{E}$ where $f(x) = P$, it forwards $\mathcal{D}^{t-1}[x]$ to $f(y)$.
When $f(y)$ receives $D^{t-1}[x]$, it merges it into its next layer local sketch for $y$, $\mathcal{D}^t[y]$, computing Equation~\eqref{eq:nextlayer} once all messages are processed.
By construction, we have that
\begin{equation} \label{eq:nextlayer:full}
\mathcal{D}^t[x] = \widetilde{\bigcup}_{y: d(x,y) = s < t-1} \mathcal{D}^{s}[y].
\end{equation}
Ergo, the set of elements inserted into $\mathcal{D}^t[x]$ consists of all $y \in \mathcal{V}$ such that $d(x,y) \leq t$, which is to say that $\widetilde{|\mathcal{D}^t[x]|}$ directly approximates $\mathcal{N}(x, t)$ (Equation~\eqref{eq:func:nbhd:local}). 
These data structures can be maintained for later use by simply storing all $\mathcal{D}^t$ between passes.

The summations over all sketches in line 
\ref{alg:ds:anf:line:sum2} of Algorithm~\ref{alg:ds:anf} estimate $\mathcal{N}(t)$ in the form Equation~\ref{eq:func:nbhd:sum}.
Note that these summations are performed as distributed \algoname{Reduce} operations, and occur between passes over $\mathbf{\sigma}$.
The following theorem states the approximation quality of Algorithm~\ref{alg:ds:anf} when implemented with HLLs and is inspired by Theorem 1 of \citep{boldi2011hyperanf}.

\newtheorem{theorem}{Theorem}
\begin{theorem} \label{thm:nbhd}
Let $\mu_{r, n}$ and $\eta_{r, n}$ be the multiplicative bias and standard deviation for \algoname{HLL}s given in Theorem 1 of \citep{flajolet2007hyperloglog}.
The output $\widetilde{\mathcal{N}}(t)$ and $\widetilde{\mathcal{N}}(x, t)$ for $x \in \mathcal{V}$ at the $t$-th iteration satisfies 
\begin{equation*}
	\frac{\E \left [ \widetilde{\mathcal{N}}(t) \right ]}{\mathcal{N}(t)} 
	= \frac{\E \left [ \widetilde{\mathcal{N}}(x, t) \right ]}{\mathcal{N}(x, t)} 
	= \mu_{r, n} 
	\textnormal{ for $n \rightarrow \infty$,}
\end{equation*}
i.e. they are nearly unbiased.
Furthermore, both also have standard deviation bounded by $\eta_{r,n}$.
That is, 
\begin{equation*}
	\frac{\sqrt{\Var \left [ \widetilde{\mathcal{N}}(t)\right ]}}{\mathcal{N}(t)} \leq \eta_{r, n}
	\textnormal{ and }
	\frac{\sqrt{\Var \left [ \widetilde{\mathcal{N}}(x, t) \right ]}}{\mathcal{N}(x, t)} \leq \eta_{r, n}
\end{equation*}
\end{theorem}

\begin{proof}
For each $x$, $\widetilde{\mathcal{N}}(x, t) = |\mathcal{D}^k[x]|$, where $\mathcal{D}^k[x]$ is a union of \algoname{HLL}s, into which every $y$ such that $d(x,y) \leq t$ is inserted, as we noted from Equation~\eqref{eq:nextlayer:full}. 
Thus by Theorem 1 of \citep{flajolet2007hyperloglog}, 
\begin{align*}
\E \left [ \widetilde{\mathcal{N}}(x, t) \right ] 
&= \mu_{r,n} \mathcal{N}(x, t) \\
\sqrt{\Var \left [ \widetilde{\mathcal{N}}(x, t) \right ]} 
&= \eta_{r,n} \mathcal{N}(x, t).
\end{align*}
Thus, by the linearity of expectation and the subadditivity of variance, 
\begin{align*}
\E \left [ \widetilde{N}(t) \right ] 
&= \sum_{x \in \mathcal{V}} \E \left [ \widetilde{\mathcal{N}}(x, t) \right ]  \\
&= \mu_{r, n} \sum_{x \in \mathcal{V}} \mathcal{N}(x, t)  
= \mu_{r,n} \mathcal{N}(t),
\textnormal{ and} \\
\sqrt{\Var \left [ \widetilde{N}(t) \right ]} 
&\leq \sum_{x \in \mathcal{V}} \sqrt{\Var \left [ \widetilde{\mathcal{N}}(x, t) \right ]}  \\
&\leq \eta_{r, n} \sum_{x \in \mathcal{V}} \mathcal{N}(x, t)  
= \eta_{r, n} \mathcal{N}(t).
\end{align*}
\end{proof}

Theorem~\ref{thm:nbhd} tells us that the estimates of $\widetilde{N}(x, t)$ and $\widetilde{N}(t)$ retain the approximation guarantees of their underlying sketch.
Similar theorems for different cardinality sketches with closed $\widetilde{\cup}$ operators are similarly simple to prove. 
Hence, we are able to guarantee that all approximations produced by Algorithm~\ref{alg:ds:anf} retain the guarantees of their underlying cardinality sketches.

\subsection{Edge-Local Triangle Count Heavy Hitters}
 \label{sec:edge_triangles}

In addition to estimating local neighborhood sizes, \algoname{DegreeSketch} affords an analysis of local triangle counts using intersection estimation. 
Furthermore, while sampling-based streaming algorithms are limited to vertex-local triangle counts, \algoname{DegreeSketch} affords the analysis of edge-local triangle counts, which 
can be thought of as a generalization of vertex-local triangle counts. 
Given the edge-local triangle counts for each edge incident upon a vertex, we can easily compute its vertex-local triangle count as in Equation~(\ref{eq:sum_of_edge_tris}).
The reverse is not true. 

Edge-local triangle counts have understandably not received much attention in the streaming literature, considering that even enumerating them requires $\Omega(m)$ space. 
Given an accumulated \algoname{DegreeSketch} $\mathcal{D}$ we can estimate $\mathcal{T}(xy)$ using the intersection estimation operator,
\begin{equation} \label{eq:tri:edge}
	\widetilde{\mathcal{T}}(xy) 
	= \widetilde{|\mathcal{D}[x] \cap \mathcal{D}[y]|}.
\end{equation}
This procedure is similar in spirit to the well-known if suboptimal intersection method for local triangle counting. 
We can also estimate the total number of triangles $\mathcal{T}$ following Equation~(\ref{eq:global_tris}) by
\begin{equation} \label{eq:tri:total}
	\widetilde{\mathcal{T}} 
	= \frac{1}{3} \sum\limits_{xy \in \mathcal{E}} 	\widetilde{\mathcal{T}}(xy) 
\end{equation}

\begin{algorithm} [t]
\caption{Local Triangle Count Heavy Hitters Chassis}\label{alg:ds:chassis}
\begin{flushleft}
        \textbf{Input:} 		$\boldsymbol{\sigma}$ - edge stream divided into $|\mathcal{P}|$ substreams\\
        	\hspace{3.2em}	$k$ - integral heavy hitter count	 \\
        	\hspace{3.2em}	$\mathcal{P}$ - universe of processors	 \\
        	\hspace{3.2em}	$\mathcal{D}$ - accumulated \algoname{DegreeSketch} 	 \\
        	\hspace{3.2em}	$\mathcal{S}$ - dictionary mapping $\mathcal{P}$ to send queues	 \\
        	\hspace{3.2em}	$\mathcal{R}$ - dictionary mapping $\mathcal{P}$ to receive queues	 \\
        	\hspace{3.2em}	$f$ - partition mapping $\mathcal{V} \rightarrow \mathcal{P}$	 \\
\end{flushleft}
\begin{flushleft}
\begin{algorithmic}[1]
	\Statex \textbf{Computation Context} for $P \in \mathcal{P}$:
		\State $\widetilde{\mathcal{H}}_k \gets $ empty max $k$-heap
		\State $\widetilde{\mathcal{T}} \gets 0$
		\While{$\sigma_P$ has unread element $uv$}
 			\State $\mathcal{S}[P].\push{\algoname{Edge}, (f(u), uv)}$
		\EndWhile
		\State \algoname{Reduce}  $\widetilde{\mathcal{T}}$
		\State $\widetilde{\mathcal{T}} \gets \frac{1}{3} \widetilde{\mathcal{T}}$
		\State \algoname{Reduce} $\widetilde{\mathcal{H}}_k$
\end{algorithmic}
\end{flushleft}
\end{algorithm}

Unfortunately, while most cardinality sketches have a native and closed $\widetilde{\cup}$ operation, they all lack a satisfactory intersection operation, a consequence of the fact that detection of a trivial intersection is impossible in sublinear memory. 
Hence, we must instead make use of unsatisfactory intersection operations in practice, which has been a focus of recent research \citep{ting2016towards, cohen2017minimal, ertl2017new}.
We will discuss these in more detail in Section~\ref{sec:HLL:intersection}, and analyze their shortcomings in Appendix~\ref{apdx:intersections}.
For our purposes, we will suppose that $\widetilde{| \cdot \cap \cdot|}$ is reliable only where intersections are large.
Consequently, we will attempt only to recovery the heavy hitters, i.e. the edges participating in the greatest number of triangles. 

Algorithm~\ref{alg:ds:chassis} provides a chassis for Algorithms~\ref{alg:ds:edge_local}  and ~\ref{alg:ds:vertex_local}, which differ only in their communication behavior. 
In Algorithm~\ref{alg:ds:chassis}, all processors read over their edge streams and forward edges to one of their endpoints, similar to the behavior in the \textbf{Accumulation Context} of Algorithm~\ref{alg:ds:anf}.
They also initialize a counter $\widetilde{\mathcal{T}}$ and a max heap with a maximum size of $k$, $\widetilde{\mathcal{H}}_k$. 
These values are modified in the \textbf{Send Context} and \textbf{Receive Context}.

\begin{algorithm}[t] 
\caption{Edge-Local Triangle Count Heavy Hitters}\label{alg:ds:edge_local}
\begin{flushleft}
        \textbf{Returns:} $\widetilde{\mathcal{T}}$, $\widetilde{\mathcal{H}}_k$, the top $k$ edge estimates $\widetilde{\mathcal{T}}(xy)$
\end{flushleft}
\begin{flushleft}
\begin{algorithmic}[1]
	\Statex \textbf{Send Context} for $P \in \mathcal{P}$:
		\While{$\mathcal{S}[P]$ is not empty}
			\If {next message is an \algoname{Edge}}
	  			\State $(f(x), xy) \gets \mathcal{S}[P].\pop$
				\State $\mathcal{R}[f(x)].\push{\algoname{Edge}, xy}$
			\ElsIf{next message is a \algoname{Sketch}}
				\State $(f(y), \mathcal{D}[x], y) \gets \mathcal{S}[P].pop()$
				\State $\mathcal{R}[f(y)].\push{\algoname{Sketch}, xy}$
			\EndIf
  		\EndWhile
	\Statex \textbf{Receive Context} for $P \in \mathcal{P}$:
		\While{$\mathcal{R}[P]$ is not empty}
			\If {next message is an \algoname{Edge}}
	  			\State $xy \gets \mathcal{R}[P].\pop$
				\State $\mathcal{S}[P].\push{\algoname{Sketch}, (f(y), \mathcal{D}[x], xy)}$
			\ElsIf{next message is a \algoname{Sketch}}
				\State $(\mathcal{D}[x], xy) \gets \mathcal{R}[P].\pop$
				\State $\widetilde{\mathcal{T}}(xy) \gets \widetilde{|\mathcal{D}[y] \cap \mathcal{D}[x]|}$
				\State $\widetilde{\mathcal{T}} \gets \widetilde{\mathcal{T}} + \widetilde{\mathcal{T}}(xy)$\
				\State Try to insert $\widetilde{\mathcal{T}}(xy)$ into $\widetilde{\mathcal{H}}_k$
			\EndIf
  		\EndWhile
	\Statex \textbf{Computation Context} for $P \in \mathcal{P}$:
		\State Run Algorithm~\ref{alg:ds:chassis} using these communication contexts
\end{algorithmic}
\end{flushleft}
\end{algorithm}

Algorithm~\ref{alg:ds:edge_local} issues a chain of messages for each read edge, not unlike the procedure in Algorithm~\ref{alg:ds:anf}.
$P$ reads $uv$, and issues a message of type \algoname{Edge} containing $uv$ to $f(u)$.
Upon receipt, $f(u)$ issues a message of type \algoname{Sketch} containing $(\mathcal{D}[u], uv)$ to $f(v)$.
When $f(v)$ receives this message, it computes $\widetilde{\mathcal{T}}(uv)$ via Equation~\eqref{eq:tri:edge} and updates $\widetilde{\mathcal{T}}$ and $\widetilde{\mathcal{H}}_k$. 
Once computation is complete and all receive queues are flushed, the algorithm computes a global \algoname{Reduce} sum to find $\widetilde{\mathcal{T}}$ and similarly finds the global top $k$ estimates via a reduce on $\widetilde{\mathcal{H}}_k$. 
The algorithm returns $\widetilde{\mathcal{T}}/3$ (each triangle is counted 3 times) and $\widetilde{\mathcal{H}}_k$.

Algorithm~\ref{alg:ds:edge_local} addresses edge--local triangle count heavy hitter recovery using memory sublinear in the size of $\mathcal{G}$. 
It requires $\widetilde{O}(\varepsilon^{-2}m)$ time and communication, given our assumptions, and a total of $O(\varepsilon^{-2} |\mathcal{V}| \log\log |\mathcal{V}| + \log |\mathcal{V}|)$ space, where DegreeSketch is implemented using $\algoname{HyperLogLog}$ sketches with accuracy parameter  $\varepsilon$. 
Unfortunately, we are unable to provide an analytic bound on the error of this algorithm, due to the nature of sublinear intersection estimation. 
We provide an experimental exploration of this problem in Appendix~\ref{apdx:intersections}, and evaluate the performance of Algorithm~\ref{alg:ds:edge_local} in Section~\ref{sec:experiments}.

\subsection{Vertex-Local Triangle Count Heavy Hitters}
 \label{sec:vertex_triangles}

Given access to a \algoname{DegreeSketch} instance $\mathcal{D}$ and the neighbors of $x$, we can compute an estimate of $\mathcal{T}(x)$ using
\begin{equation} \label{eq:tri:vertex}
	\widetilde{\mathcal{T}}(x) 
	= \frac{1}{2} \sum_{y: xy \in \mathcal{E}} \widetilde{\mathcal{T}}(xy)  
	= \frac{1}{2} \sum_{y: xy \in \mathcal{E}} \widetilde{|\mathcal{D}[x] \cap \mathcal{D}[y]|},
\end{equation}
following Equations~(\ref{eq:sum_of_edge_tris}) and ~(\ref{eq:tri:edge}).
We must still limit our scope to the recovery of vertex-local triangle count heavy hitters due to problem of estimating small intersections.

Algorithm~\ref{alg:ds:vertex_local} performs vertex-local triangle count estimation in a manner similar to Algorithm~\ref{alg:ds:edge_local} with some additional steps. 
We maintain $\widetilde{\mathcal{T}}(x)$ for each $x \in \mathcal{V}$.
It performs similar work for $xy \in \mathcal{E}$ up to the point processor $f(y)$ estimates $\widetilde{\mathcal{T}}(xy)$. 
Instead of inserting this estimate into a local max heap, we add it to $\widetilde{\mathcal{T}}(y)$, and forward $(\widetilde{\mathcal{T}}(xy), x)$ to $f(x)$ so that it can add it to $\widetilde{\mathcal{T}}(x)$.
This message has the \algoname{Est} type, to distinguish it from \algoname{Edge} and \algoname{Sketch} messages.
After all processors are done communicating and updating their local counts, they assemble and reduce a max $k$-heap of vertex-local triangle count heavy hitters.
Note that we could also return $\widetilde{\mathcal{T}}(x)$ for all vertices $x$ at no additional cost, but this is not generally a good idea as explained in Appendix~\ref{apdx:intersections}.

Algorithm~\ref{alg:ds:vertex_local} addresses vertex--local triangle count heavy hitter recovery using the same asymptotic computation, memory and communication costs as Algorithm~\ref{alg:ds:edge_local}.
Unfortunately, we are similarly unable to provide an a priori analytic bound on the error of this algorithm. 

\begin{algorithm}[t] 
\caption{Vertex-Local Triangle Count Heavy Hitters}\label{alg:ds:vertex_local}
\begin{flushleft}
        \textbf{Output:} $\widetilde{\mathcal{T}}$, $\widetilde{\mathcal{H}}_k$, the top $k$ vertex estimates $\widetilde{\mathcal{T}}(x)$
\end{flushleft}
\begin{flushleft}
\begin{algorithmic}[1]
	\Statex \textbf{Send Context} for $P \in \mathcal{P}$:
		\While{$\mathcal{S}[P]$ is not empty}
			\If {next message is an \algoname{Edge}}
	  			\State $(f(x), xy) \gets \mathcal{S}[P].\pop$
				\State $\mathcal{R}[f(x)].\push{\algoname{Edge}, xy}$
			\ElsIf{next message is a \algoname{Sketch}}
				\State $(f(y), \mathcal{D}[x], xy) \gets \mathcal{S}[P].pop()$
				\State $\mathcal{R}[f(y)].\push{\algoname{Sketch}, (\mathcal{D}[x], xy)}$
			\ElsIf{next message is an \algoname{Est}}
				\State $(f(y), \widetilde{\mathcal{T}}(xy), y) \gets \mathcal{S}[P].pop()$
				\State $\mathcal{R}[f(y)].\push{\algoname{Est}, (\widetilde{\mathcal{T}}(xy), y)}$
			\EndIf
  		\EndWhile
	\Statex \textbf{Receive Context} for $P \in \mathcal{P}$:
		\While{$\mathcal{R}[P]$ is not empty}
			\If {next message is an \algoname{Edge}}
	  			\State $xy \gets \mathcal{R}[P].\pop$
				\State $\mathcal{S}[P].\push{\algoname{Sketch}, (f(y), \mathcal{D}[x], xy))}$
			\ElsIf{next message is a \algoname{Sketch}}
				\State $(\mathcal{D}[x], xy) \gets \mathcal{R}[P].\pop$
				\State $\widetilde{\mathcal{T}}(xy) \gets \widetilde{|\mathcal{D}[y] \cap \mathcal{D}[x])|}$
				\State $\widetilde{\mathcal{T}}(x) \gets \widetilde{\mathcal{T}}(x) + \widetilde{\mathcal{T}}(xy)$
				\State $\widetilde{\mathcal{T}} \gets \widetilde{\mathcal{T}} + \widetilde{\mathcal{T}}(xy)$
				\State $\mathcal{S}[P].\push{\algoname{Est}, (f(y), \widetilde{\mathcal{T}}(xy), y)}$
			\ElsIf{next message is an \algoname{Est}}
				\State $(\widetilde{\mathcal{T}}(xy), y) \gets \mathcal{R}[P].\pop$
				\State $\widetilde{\mathcal{T}}(y) \gets \widetilde{\mathcal{T}}(y) + \widetilde{\mathcal{T}}(xy)$
			\EndIf
  		\EndWhile
	\Statex \textbf{Computation Context} for $P \in \mathcal{P}$:
		\State $\widetilde{\mathcal{T}}(x) \gets 0$ $\forall x \in f^{-1}(P)$
		\State Run Algorithm~\ref{alg:ds:chassis} using these communication contexts
		\State Accumulate max heap $\widetilde{\mathcal{H}}_k$ from $\widetilde{\mathcal{T}}(x)$ $\forall x \in f^{-1}(P)$
\end{algorithmic}
\end{flushleft}
\end{algorithm}

\section{The HyperLogLog Sketch}
 \label{sec:HLL}

The HyperLogLog sketch is arguably the most popular cardinality sketch in applications, and has attained widespread adoption \citep{flajolet2007hyperloglog}.
The sketch relies on the key insight that the binary representation of a random machine word  starts with $0^{j-1} 1$ with probability $2^{-j}$. 
Thus, if the maximum number of leading zeros in a set of random words is $j-1$, then $2^j$ is a good estimate of the cardinality of the set \citep{flajolet1985probabilistic}.
However, this estimator clearly has high variance. 
The variance is traditionally minimized using stochastic averaging to simulate parallel random trials \citep{flajolet1985probabilistic}.

Assume we have a stream $\sigma$ of random machine words of a fixed size $W$.
For a $W=(p + q)$-bit word $w$, let $\xi(w)$ be the first $p$ bits of $w$, and let $\rho(w)$ be the number of leading zeros plus one of its remaining $q$ bits.
We pseudorandomly partition elements $e$ of $\sigma$ into $r = 2^p$ substreams of the form $\sigma_i = \{ e \in \sigma | \xi(e) = i \}$.
For each of these approximately equally-sized streams, we maintain an independent estimator of the above form.
Each register $\mathbf{r}_i$, $i \in [r]$, accumulates the value
\begin{equation} \label{eq:register}
\mathbf{r}_i = \max\limits_{x \in \sigma_i} \rho(x).
\end{equation}
Of course, in practice we simulate randomness using hash functions.
We utilize the non-cryptographic xxhash \citep{xxhash}  in our implementation. 
We will assume throughout that algorithms have access to such a hash function $h : 2^{64} \rightarrow 2^{64}$.

After accumulation, $\mathbf{r}_i$ stores the maximum number of leading zeroes in the substream $\sigma_i$, plus one. 
The authors of HyperLogLog show in \citep{flajolet2007hyperloglog} that the normalized bias corrected harmonic mean of these registers,
\begin{equation} \label{eq:estimator}
\widetilde{D} = \alpha_r r^2 \left ( \sum_{i=0}^{r-1} 2^{-\mathbf{r}_i} \right) ^{-1},
\end{equation}
where the bias correction term $\alpha_r$ is given by
\begin{equation} \label{eq:alpha}
\alpha_r :=  \left( r \int_{0}^\infty \left( \log_2 \left( \frac{2 + u}{1 + u} \right) \right)^r du \right) ^{-1},
\end{equation}
is a good estimator of $D$, the number of unique elements in $\sigma$.
The error of estimate $\widetilde{D}$, $|D - \widetilde{D}|$, has standard error $\approx 1.04 / \sqrt{r}$, i.e. $\widetilde{D}$ satisfies
\begin{equation} \label{eq:error_bound}
|D - \widetilde{D}| \leq (1.04/\sqrt{r}) D.
\end{equation}
 with high probability.
 However, Equation~\ref{eq:estimator} is known to experience practical bias on small and very large $D$.
 Expanding the hash function to handle 64-bit words instead of the original 32-bit words practically eliminates bias on large $D$ in most practical problems.
 Although modifications to handle small $D$ bias abound, we choose the approach of \algoname{LogLogBeta} \citep{qin2016loglog} for its simplicity and replace $\widetilde{D}$ with  
\begin{equation} \label{eq:beta_estimator}
\widetilde{E} = \alpha_r r (r - z) \left ( \beta(r, z) + \sum_{i=0}^{r-1} 2^{-\mathbf{r}_i} \right) ^{-1}.
\end{equation}
Here $z$ is the number of zero registers in $r$ and $\beta(r, z)$ is an experimentally determined bias minimizer.
We follow the lead of the authors of \algoname{LogLogBeta} and determined $\beta(r, z)$ as a 7th-degree polynomial of $\log(z)$, whose weights are set experimentally by solving a least-squares problem like in Section II.C. of \citep{qin2016loglog}. 

The majority of vertices in many application graphs have a small number of neighbors, which suggests that to maximize memory and communication efficiency we would like a cheaper way to represent sketches where most of $\mathbf{r}$ is empty.
We adopt the sparse representation for \algoname{HyperLogLog} sketches suggested by Heule et al. \citep{heule2013hyperloglog}.
Mathematically, the sparsification procedure is tantamount to maintaining the set $R = \left \{(i, \mathbf{r}_i) | \mathbf{r}_i \neq 0 \right \}$.
$R$ requires less memory than $\mathbf{r}$ when the cardinality of the underlying multiset is small. 
Moreover, it is straightforward to saturate a sparse sketch into a dense one once it is no longer cost effective to maintain it by instantiating $\mathbf{r}$ while assuming all registers not set in $R$ are zero. 
We will assume that $R$ is implemented as a map, where an element $R[j] = x$ if $(j, x) \in R$ and is zero otherwise.

A particular HyperLogLog sketch, $S$, consists a hash function $h$, a prefix size $p$ (typically between 4 and 16), a maximum register value $q$, and an array of $r=2^p$ registers, initially an empty sparse register list $R$. 
We summarize references to such a sketch as $\algoname{HLL}(p,q,h)$.
Algorithm~\ref{alg:hll} describes the functions supported by our version of the \algoname{HyperLogLog} sketch.

\begin{algorithm}[t] 
\caption{$\algoname{HLL}(p,q,h)$ Operations}\label{alg:hll}
\begin{flushleft}
        \textbf{State Variables} for $\algoname{HLL}(p, q, h)$ $S$ :		\\
    		\hspace{2.65em}	$\nu$  \,\,\, mode $\in \{\algoname{sparse}, \algoname{dense}\}$, initially $\algoname{sparse}$ \\
    		\hspace{2.65em}	$r$  \,\,\,\, $2^p$ \\
    		\hspace{2.65em}	$\mathbf{r}$  \,\,\,\, dense register list of size $r$, initially $\emptyset$ \\
    		\hspace{2.65em}	$R$  \,\,\, sparse register set, initially $\emptyset$ \\
    		\hspace{2.65em}	$z$  \,\,\,\, count of nonzero elements of $R$ or $\mathbf{r}$, initially $2^p$
\end{flushleft}
\begin{algorithmic}[1]
	\Function{Insert}{$S, e$}
		\State $w \gets h(e)$
		\State $j \gets \xi(w)$
		\State $x \gets \rho(w)$
		\State $\algoname{Insert}(S, j, x)$
	\EndFunction
	\Function{Insert}{$S, j, x$}
		\If {$\nu = \algoname{dense}$}
			\State $\mathbf{r}_j \gets \max ( \mathbf{r}_j, x)$
		\ElsIf {$\nu = \algoname{sparse}$}
			\State $R[j] \gets \max \{x, R[j]\}$ (see Figures 6 \& 7 of \citep{heule2013hyperloglog})
			\If {$|R| > r / 4$}
				\State $\algoname{Saturate}(S)$
			\EndIf
		\EndIf
	\EndFunction
	\Function{Saturate}{$S$}
		\State $\nu \gets \algoname{dense}$
		\For {$(j, x) \in R$} 
			\State $\algoname{Insert}(S, j, x)$
		\EndFor
		\State $R \gets \emptyset$
	\EndFunction
	\Function{Merge}{$S^{(0)}, S^{(1)}, \dots, S^{(\ell-1)}$}
		\State $S^* \gets \text{empty $\algoname{HLL}(p, q, h)$}$
		\For{$j \in [0,r)$} 
			\State $x \gets  \max\limits_{ i \in [0, \ell)} \left ( \max \left \{ \mathbf{r}^{(i)}_j, R^{(i)}[j] \right \} \right )$
			\If {$x \neq 0$} 
				\State $\algoname{Insert}(S^*, j, x)$
			\EndIf
		\EndFor
		\State \Return $S^*$
	\EndFunction
	\Function{Estimate}{$S$}
		\If{$\nu = \algoname{dense}$}
			\State \Return $\alpha_r r (r - z) \left ( \beta(r, z) + \sum_{i=0}^{r-1} 2^{-\mathbf{r}_i} \right) ^{-1}$
		\Else
			\State \Return $\alpha_r r (r - z) \left ( \beta(r, z) + \sum_{(j, x) \in R} 2^{-x} \right) ^{-1}$
		\EndIf
	\EndFunction
\end{algorithmic}
\end{algorithm}

Note that \algoname{HLL}s support a natural merge operation: taking the element-wise maximum of each index of a group of register vectors.
This requires that the two sketches were generated using the same hash function.
We assume throughout that all of the sketches in an instance of \algoname{DegreeSketch} are $\algoname{HLL}(p, q, h)$ where $p + q = 64$ and $h$ is fixed.
Estimate accuracy scales inverse squared with $p$ per Equation~(\ref{eq:error_bound}).
Thus, increasing $p$ improves estimation performance at the cost of computation and memory overhead.

\subsection{Intersection Estimation}
 \label{sec:HLL:intersection}

A na\"ive approach to estimating an intersection of two sets $A$ and $B$ using cardinality sketches might involve computing the intersection via the inclusion-exclusion principle:
\begin{equation} \label{eq:inclusion-exclusion}
	|A \cap B| = |A \cup B| - |A| - |B|.
\end{equation}
However, if we attempt to estimate each quantity on the right side of Equation~(\ref{eq:inclusion-exclusion})  the error noise in each estimate could result in a negative answer!
Furthermore, if the true intersection is small relative to the set sizes, or if one set is much larger than the other, the variance will be quite high.

Ertl describes a better intersection estimator using a maximum likelihood principle \citep{ertl2017new}.
The estimator yields estimates of $|A \setminus B|$, $|B \setminus A|$, and $|A \cap B|$. 
The algorithm depends on the optimization of a Poisson model, where it is assumed that $|A \setminus B|$ is drawn from a Poisson distribution with parameter $\lambda_a$, and similarly $|B \setminus A|$ and $|A \cap B|$ use Poisson parameters $\lambda_b$ and $\lambda_x$. These parameters can be related to the observed HyperLogLog register lists corresonding to $A$ and $B$, $\mathbf{r}^{(A)}$ and $\mathbf{r}^{(B)}$, via a loglikelihood function $\mathcal{L}(\lambda_a, \lambda_b, \lambda_x \mid \mathbf{r}^{(A)}, \mathbf{r}^{(B)})$ given by Equation~(70) of \citep{ertl2017new}, which we do not reproduce due to space constraints.
$\mathcal{L}(\lambda_a, \lambda_b, \lambda_x \mid \mathbf{r}^{(A)}, \mathbf{r}^{(B)})$ is a function of the statistics:
\begin{equation} \label{eq:cs}
	\begin{alignedat}{2}
		&\mathbf{c}^{(A),<}_k &= |\{i \mid k = \mathbf{r}^{(A)}_i < \mathbf{r}^{(B)}_i \}|, \\
		&\mathbf{c}^{(A),>}_k &= |\{i \mid k = \mathbf{r}^{(A)}_i > \mathbf{r}^{(B)}_i \}|, \\
		&\mathbf{c}^{(B),<}_k &= |\{i \mid k = \mathbf{r}^{(B)}_i < \mathbf{r}^{(A)}_i \}|, \\ 
		&\mathbf{c}^{(B),>}_k &= |\{i \mid k = \mathbf{r}^{(B)}_i > \mathbf{r}^{(A)}_i \}|, \\
		&\mathbf{c}^{=}_k &= |\{i \mid k = \mathbf{r}^{(A)}_i = \mathbf{r}^{(B)}_i \}|,
  \end{alignedat}
\end{equation} 
which capture the differences in register list distribution.
The inclusion-exclusion estimator 
loses information present in the more detailed count statistics statistics in Equation~\eqref{eq:cs}.
Algorithm 9 of \citep{ertl2017new} describes the estimation of $|A \setminus B|$, $|B \setminus A|$, and $|A \cap B|$ by accumulating the sufficient statistic \eqref{eq:cs} and using it to find the maximum of Equation~(70) in the source via maximum likelihood estimation.
The author shows extensive simulation evidence indicating that this method significantly improves upon the estimation error of a na\"ive estimator.
We provide further analyses of intersection estimation edge cases in Appendix~\ref{apdx:intersections}, and reaffirm some of Ertl's findings.

\section{Experiments}
 \label{sec:experiments}

We now evaluate the algorithms and claims made throughout this document.

\noindent
\textbf{Implementation}:
We implemented all of our algorithms in C++ and MVAPICH2 2.3. 
Inter- and intra-node communication is managed using the pseudo-asynchronous MPI-enabled communication software package YGM \citep{priest2019you}.
We used xxhash as our hash function implementation \citep{xxhash}.

\begin{figure}
\centering
\includegraphics[width=0.7\columnwidth]{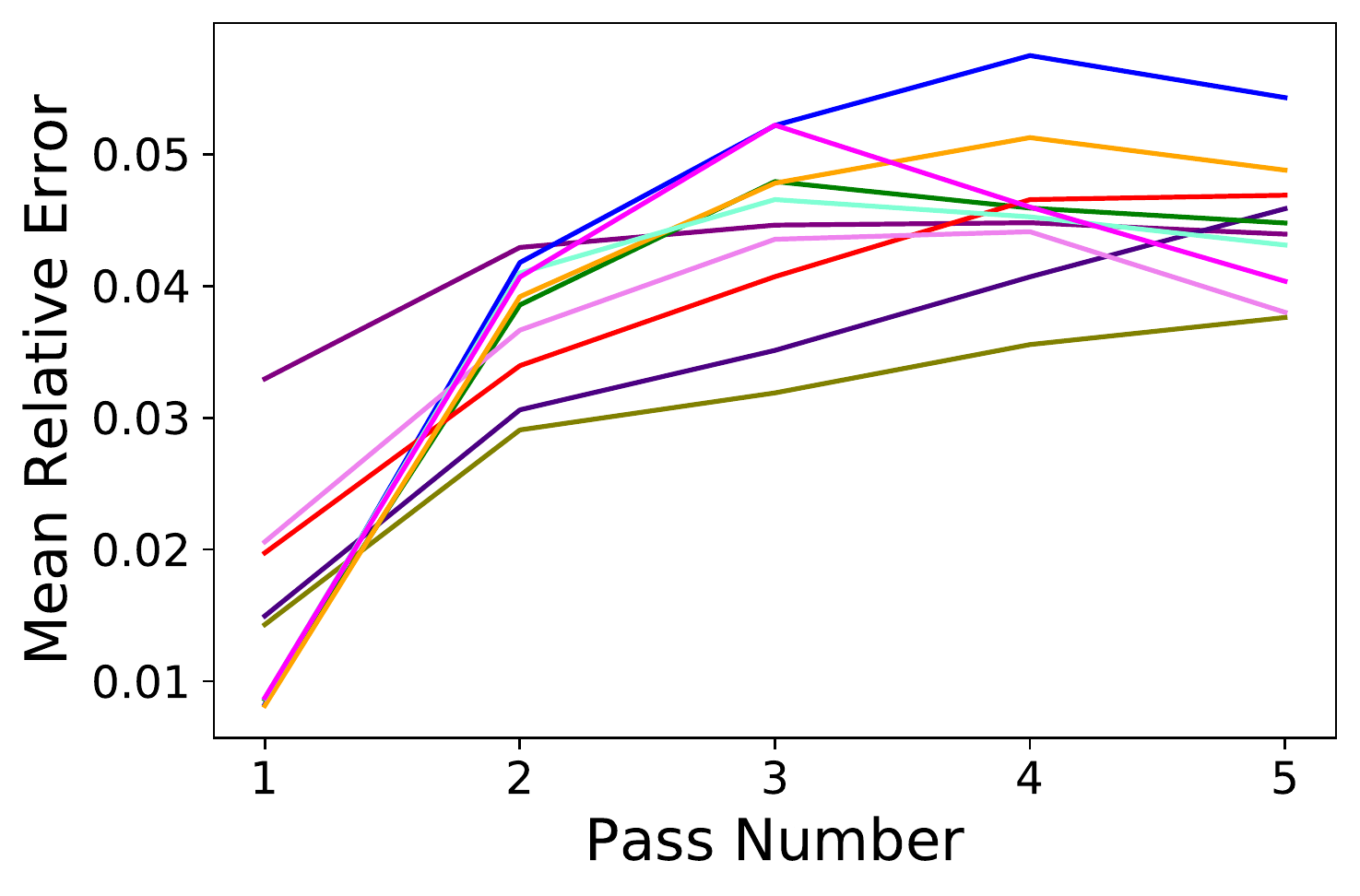}
\caption{Mean relative error estimating $\mathcal{T}(x, t)$ for all $x \in \mathcal{V}$ and $t$ up to 5 for 10 moderate graphs.
	Here a prefix size of 8 was used, guaranteeing standard error around 0.06.
}
\label{fig:nbhd}
\end{figure}

\noindent
\textbf{Hardware}: 
All of the experiments were performed on a cluster of compute nodes with twenty-four 2.4 GHz Intel Xeon E5-2695 v2 processor cores and 128GB memory per node. 
We varied the number of nodes per experiment depending on scalability requirements and the size of the graph.
We consider graph partitioning to be a separate problem, and accordingly use simple round-robin assignment for our experiments.

\noindent
\textbf{Graphs}: 
We ran our implementations on a collection of real and synthetic graphs. 
Many of these graphs are provided by Stanford's widely used SNAP dataset \citep{snapnets}.
These graphs are collected from natural sources, such as social media, transportation networks, email records, peer-to-peer communications, and so on.
We casted each graph as unweighted, ignoring directionality, self-loops, and repeated edges. 
We also used 5 graphs derived from nonstochastic Kronecker products of smaller graphs.
These graphs are described in detail in Appendix~\ref{apdx:kronecker}.

\noindent
\textbf{Experiments}:
Except where noted otherwise, we ran each experiment 100 times using the same settings while varying the random seed.
We set the prefix size $p$ in experiments depending on the accuracy needs, where larger $p$ implies greater performance at a higher cost.
We report mean relative error (MRE), where the relative error of and estimate $E$ of a true quantity $T$ is $\frac{|T - E|}{|T|}$.

\noindent
\textbf{Analysis}: 
We performed experiments on real graph datasets for the purpose of establishing the following of our algorithms:
\begin{enumerate}
	\item \textbf{Estimation Quality} Do the algorithms yield good local $t$-neighborhood estimates? 
		How accurate are the global and edge- and vertex-local estimates? 
	\item \textbf{Heavy Hitter recovery	} Do the heavy hitters returned by Algorithms~\ref{alg:ds:edge_local} and \ref{alg:ds:vertex_local} correspond to the ground truth heavy hitters?
	\item \textbf{Speed \& Scalability} How fast is accumulation? Estimation? How does wall time relate to $|\mathcal{P}|$?
\end{enumerate}

\begin{figure*}[t]
\centering
\includegraphics[width=1.0\columnwidth]{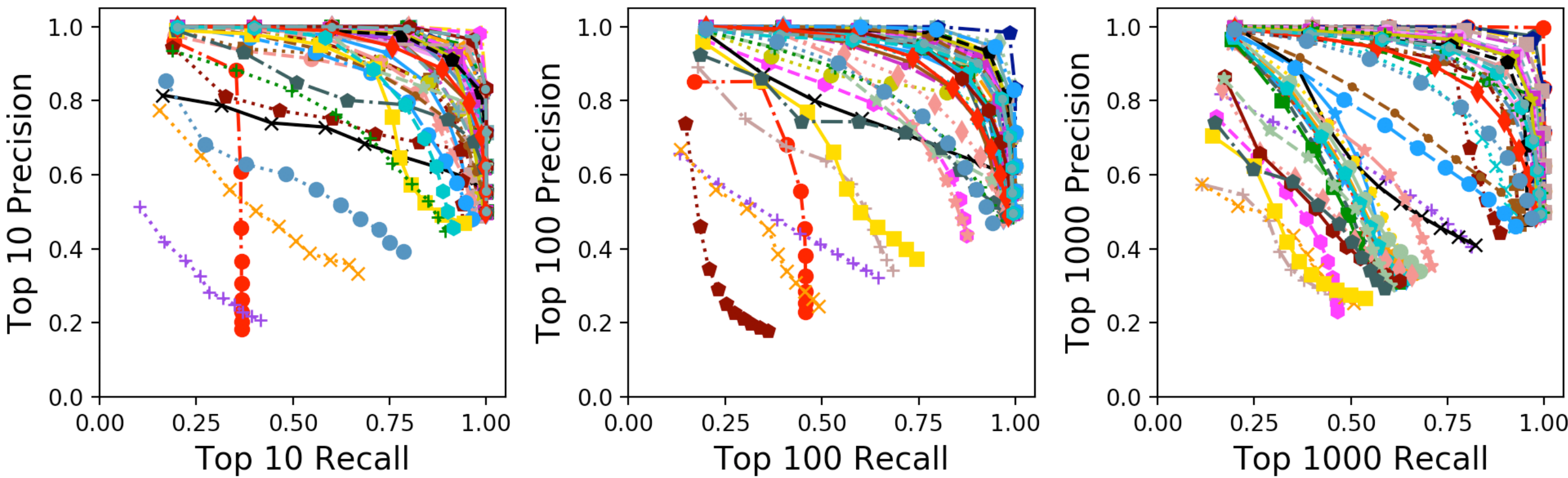}
\caption{Precision versus recall for the top 10, top 100, and top 1000 ground edge-local triangle count truth heavy hitters of all SNAP graphs and all 5 synthetic kronecker graphs.
}
\label{fig:precision_recall}
\end{figure*}

We examined the performance of the local $t$-neighborhood estimation algorithm (Algorithm~\ref{alg:ds:anf}) 
with prefix size of 8 on 10 moderately sized SNAP graphs up to $t = 5$.
Figure~\ref{fig:nbhd} displays the MRE of the returned local estimates. 
We find that the MRE matches our expectations informed by Theorem~\ref{thm:nbhd}. 
In early passes, most of the neighborhoods are relatively small and so in practice the cardinality sketches give very precise estimates. 
As $t$ grows, the neighborhood balls grow to saturate the graph and so accordingly the estimation error grows until leveling off around the theoretical error guarantee.

\begin{figure}[t]
\centering
\includegraphics[width=0.9\columnwidth]{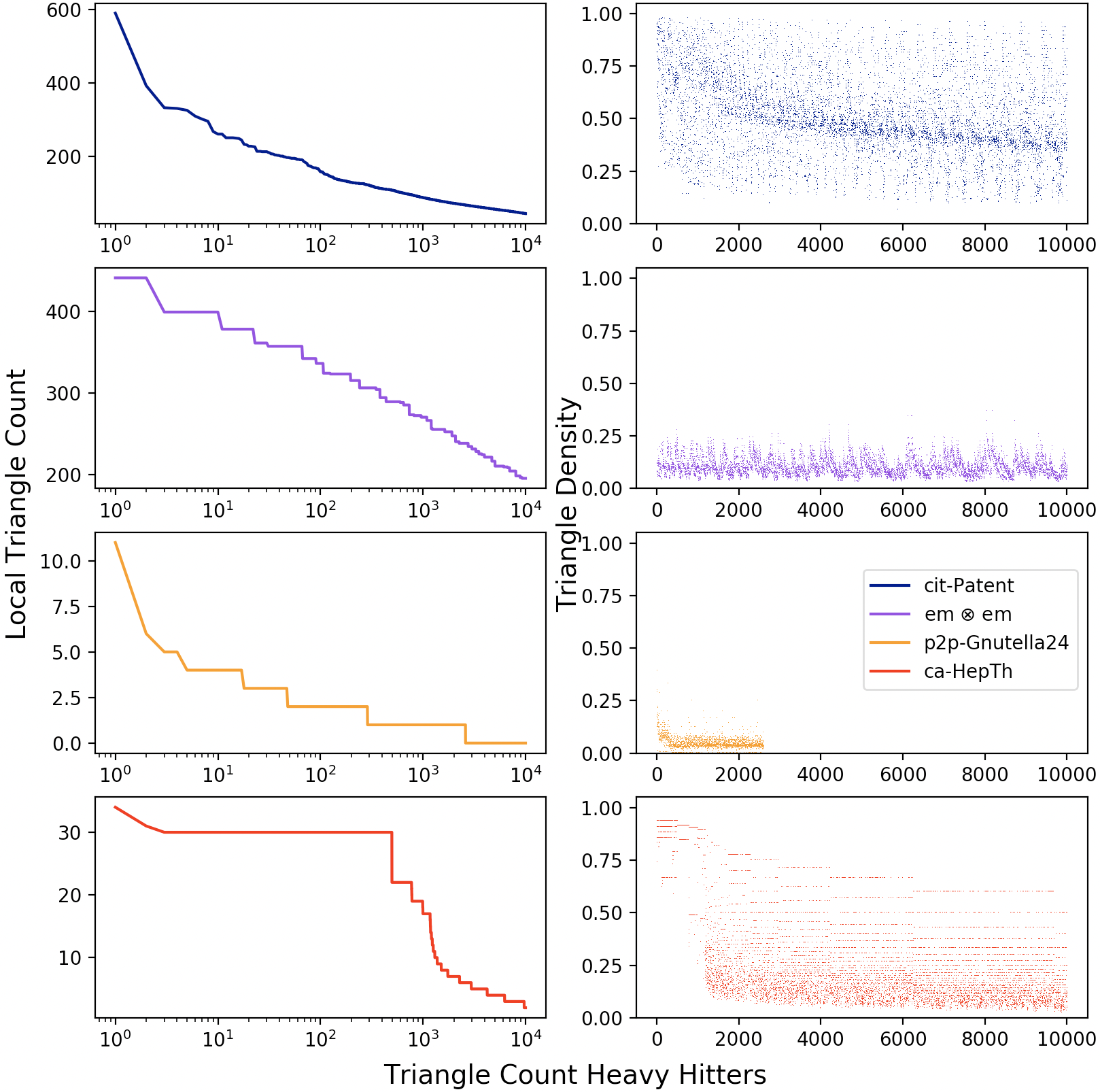}
\caption{The triangle counts and triangle densities of the edge-local triangle count heavy hitters up to $10^4$ for four graphs.
}
\label{fig:hh_comp}
\end{figure}

We experimented with Algorithm~\ref{alg:ds:edge_local} with a prefix size of 12, attempting to collect the top $k = 10, 100, 1000$ heavy hitters $\mathcal{H}_k$ for all of the SNAP and Kronecker graphs.
For each $k$, we ran the algorithm with $k^\prime$ ranging from $0.2 * k$ up to $2 * k$, producing $\widetilde{\mathcal{H}}_{k^\prime}$.
We performed a similar analysis and found similar results for Algorithm~\ref{alg:ds:vertex_local}, and omit them for space.

Treating $\widetilde{\mathcal{H}}_{k^\prime}$ as a one-class classifier of elements in $\mathcal{H}_k$, an edge $e \in \mathcal{E}$ is a true positive if $e \in \mathcal{H}_{k} \cap \widetilde{H}_{k^\prime}$, a false negative if $e \in \mathcal{H}_{k} \setminus \widetilde{\mathcal{H}}_{k^\prime}$, a false positive if $e \in \widetilde{\mathcal{H}}_{k^\prime} \setminus \mathcal{H}_{k}$, or a true negative if $e \not\in \mathcal{H}_{k} \cup\widetilde{\mathcal{H}}_{k^\prime}$.

We can report the quality of experiments in terms of its recall $\left ( \frac{TP}{TP + FN} \right )$ versus its precision $\left ( \frac{TP}{TP + FP} \right )$ for each setting of $k$ and $k^\prime$ in Figure~\ref{fig:hh_comp}.
The precision versus recall tradeoff is a common metric in information retrieval, where the goal is to tune model parameters so as to force both the precision and recall as close to one as possible.
Although these measures are known to exhibit bias, they are accepted as being reasonable for heavily uneven classification problems such as ours, where the class of interest is a small proportion of the samples \citep{boughorbel2017optimal}. 
Although most graphs show very good precision versus recall curves, there are a few notable outliers. 

Figure~\ref{fig:hh_comp} contrasts the edge-local triangle count distribution between a graph with good performance in Figure~\ref{fig:precision_recall} and three with relatively poor performance.
We find that \emph{triangle density}, or the ratio of edge-local triangles versus the union of endpoint adjacency sets is a powerful determiner of performance.
Triangle density corresponds with the Jaccard similarity of the edge's endpoint neighbor sets, and for neighbor sets $A$ and $B$ is computed as $\frac{|A \cap B|}{|A \cup B|}$.
In other words, again, relatively small intersections produce high error and uncertain heavy hitter recovery. 

The cit-Patent graph exhibits good performance in Figure~\ref{fig:precision_recall}, and demonstrates a reasonable triangle count distribution as well as high triangle density throughout in Figure~\ref{fig:hh_comp}.
The other three graphs demonstrate poor performance in Figure~\ref{fig:precision_recall}.
The kronecker em $\otimes$ em graph exhibits an unusual number of ties in its triangle count distribution due to its construction, in addition to low triangle density among its heavy hitters.
The P2P-Gnutella24 graph has very low triangle density, and a 3 or fewer triangles for the vast majority of its edges.
The ca-HepTh graph exhibits an unusual triangle distribution, where a huge portion of its edges tie at 30 triangles.
Consequently, even a perfect heavy hitter extraction procedure will fail on this graph.
Notably, the two edges with the largest triangle counts are reliably returned.

\begin{table}
\centering
\caption{Scaling Graphs \label{tab:scaling_graphs}}
\begin{tabular}{|c|c|c|c|}
\hline
\textbf{graph} & $|\mathcal{V}|$ &  $|\mathcal{E}|$ & Type \\
\hline
\hline
patents & 3,774,768 & 16,518,947 & Citation \\
\hline
\textbf{ye $\boldsymbol{\otimes}$ ye} & 5,574,320 & 88,338,632 & Kronecker \\
\hline
\textbf{or $\boldsymbol{\otimes}$ or} & 131,859,288 & 1,095,962,562 & Kronecker \\
\hline
Twitter & 41,652,224 & 1,201,045,942 & S. N. \citep{kunegis2013konect} \\
\hline
WDC & 3,563,602,788 & 128,736,914,864 & Web \\
\hline
\end{tabular}
\end{table}

\begin{figure}[t]
\centering
\includegraphics[width=0.5\columnwidth]{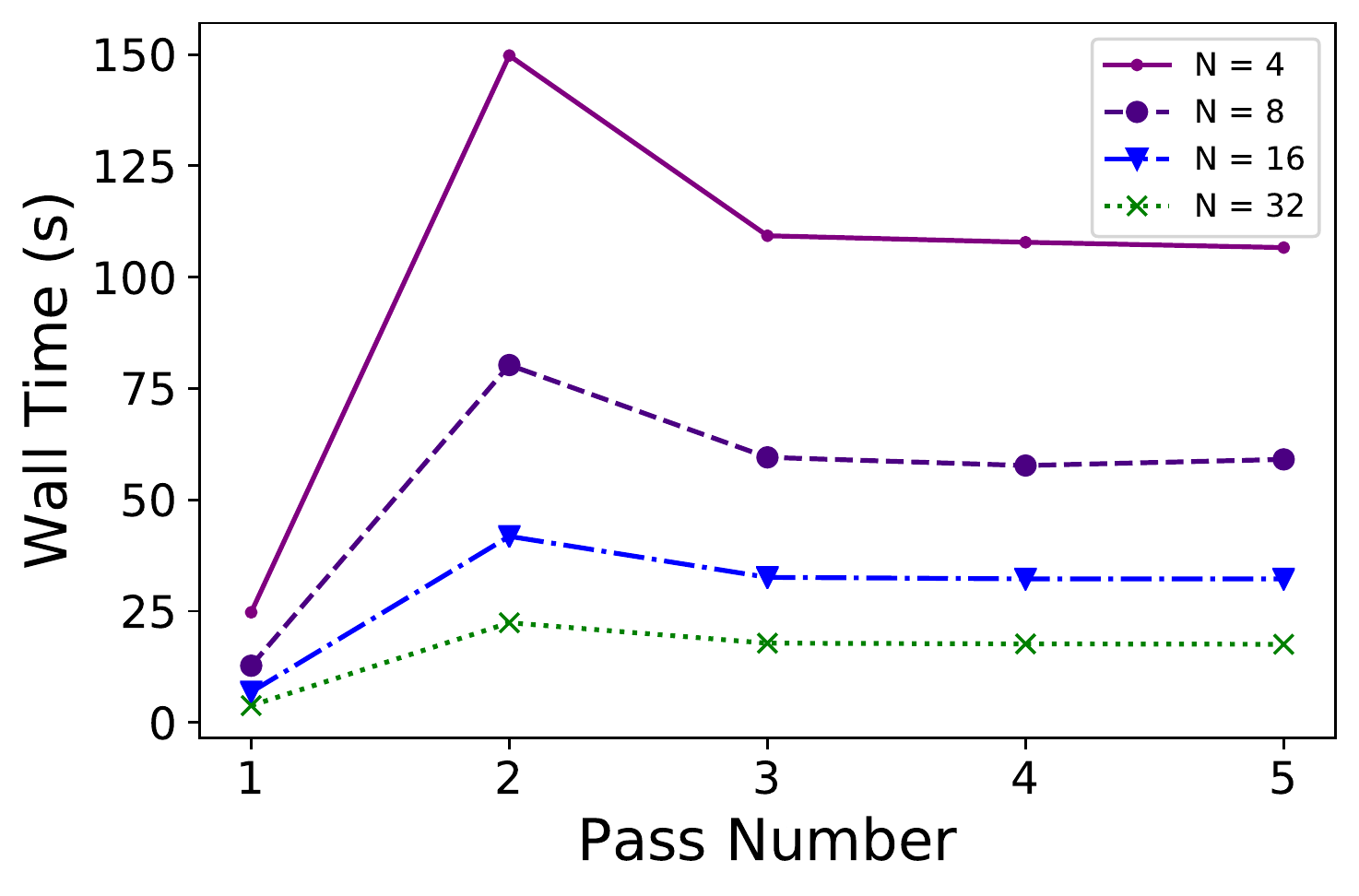}
\caption{The time in seconds to perform local $t$-neighborhood size estimation up to $t=5$ for the or $\otimes$ or graph for nodes $N = 4, 8, 16. 32$. 
}
\label{fig:nbhd_scaling}
\end{figure}

\begin{figure}
\centering
\includegraphics[width=0.7\columnwidth]{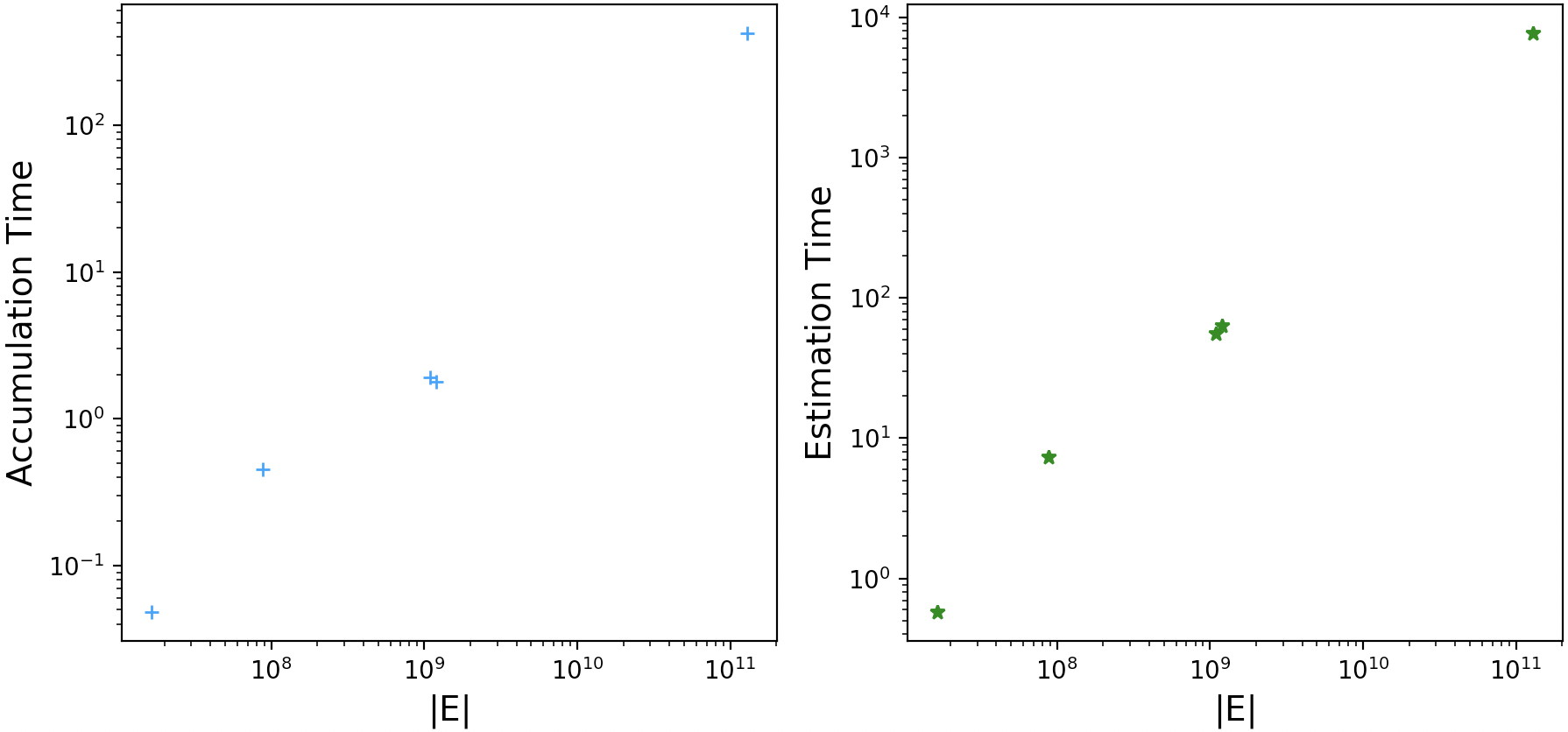}
\caption{The time in seconds to accumulate and perform local triangle count estimation using $N=72$ compute nodes for all graphs listed in Table~\ref{tab:scaling_graphs}.
}
\label{fig:strong_scaling}
\end{figure}

\begin{figure}
\centering
\includegraphics[width=0.7\columnwidth]{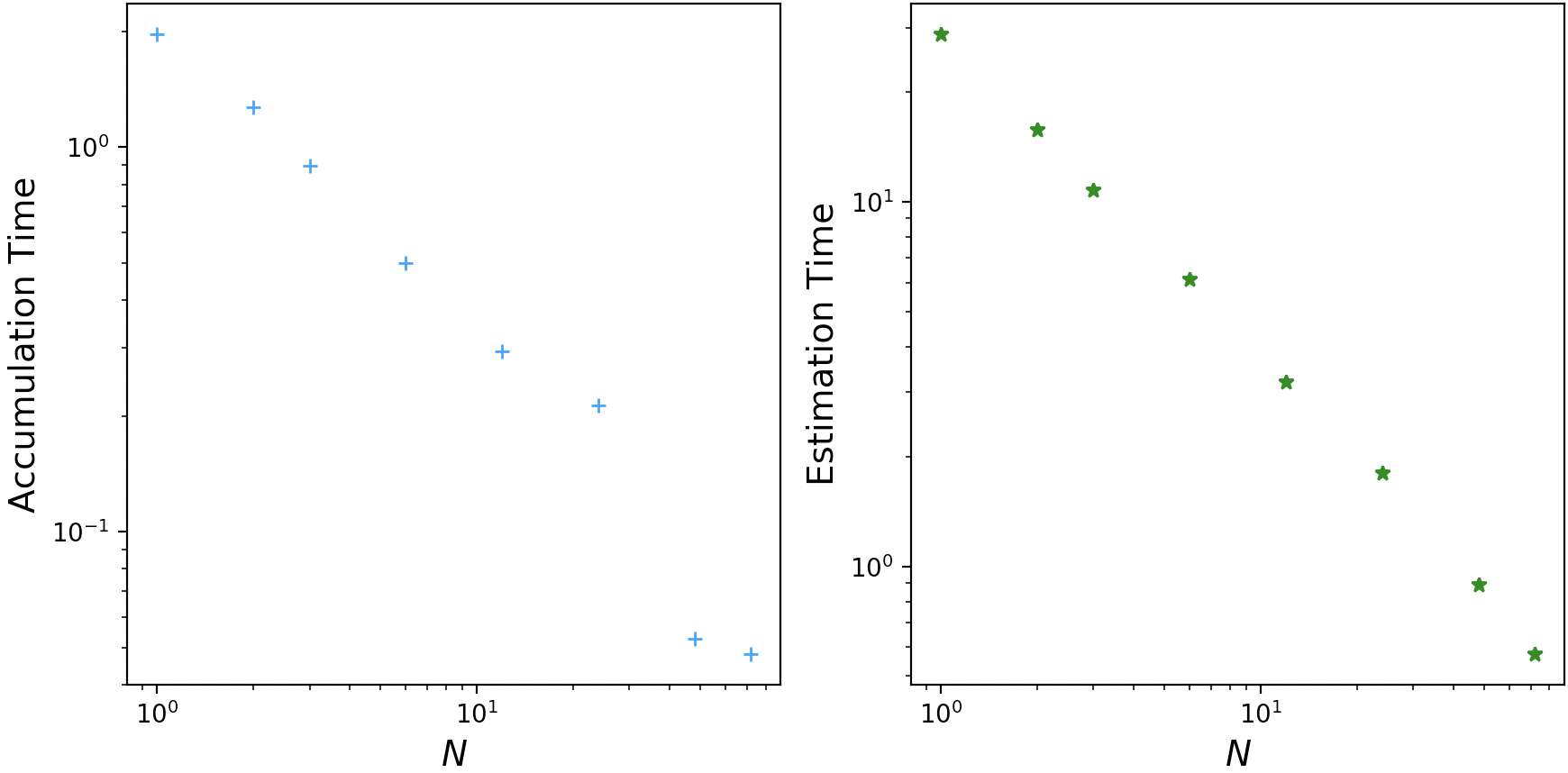}
\caption{The time in seconds to accumulate and perform local triangle count estimation using $N=1$ up to $72$ for the citation-Patents graph. 
}
\label{fig:weak_scaling}
\end{figure}

We also examined the performance scaling of the algorithms as a function of data and computing resource sizes.
For each experiment we used a prefix size of 8 and discounted the I/O time spent on reading data streams from files.
Algorithms~\ref{alg:ds:edge_local} and ~\ref{alg:ds:vertex_local} exhibit near-identical time performance, and so we report only the later.

We ran Algorithm~\ref{alg:ds:anf} for $t=5$ on the or $\otimes$ or Kronecker graph, varying the number of nodes from $N= 4, 8, 16, 32$. 
We note nice weak scaling behavior: as the computational resources double, the time roughly halves. 
In particular, pass 2 tends to take more time because of the sparsity settings in our implementations; merges are less efficient for sparse sketches. 
Once the sketches saturate, note that the time decreases significantly in later passes.
One could implement \algoname{DegreeSketch} using only dense sketches to avoid this hump.
If one only intends to perform local $t$=neighborhood estimation, omitting sparse sketches is a good idea as all sketches will eventually saturate as $t$ grows.

Similarly, Figure~\ref{fig:weak_scaling} measures the the time in seconds spent for Algorithms~\ref{alg:ds:accumulation} and ~\ref{alg:ds:vertex_local}  on the cit-Patents graph, where $N$ varies from 1 up to 72. 
This weak scaling experiment demonstrates significant performance improvements on a fixed amount of work as resources increase.

It is difficult to demonstrate strong scaling using graph data, as graphs (especially realistic ones) do not scale smoothly like, say, linear algebra.
We instead present a "strong scaling" experiment on the 5 large graphs in Table~\ref{tab:scaling_graphs}.
We found that subsequent passes of Algorithm~\ref{alg:ds:anf} exhibited similar behavior to Algorithm~\ref{alg:ds:vertex_local}, and so we only report results for the latter.

Figure~\ref{fig:strong_scaling} measures the the time in seconds spent accumulating $\algoname{DegreeSketch}$ and performing the vertex-local estimation on each graph, plotted against the number of edges in each graph. 
We used $N=72$ compute nodes in each case. 
As promised, the wall time is linear in the number edges for both accumulation and estimation, showing good scaling with graph size on fixed resources.
We found in experiments that the subsequent passes of Algorithm~\ref{alg:ds:anf} experienced similar linear scaling.

It is worth noting that a competing state-of-the-art exact triangle counting algorithm required $N=256$ compute nodes to even load the largest WebDataCommons graph into distributed memory \citep{pearce2017triangle}.

\section{Conclusions}

We have herein demonstrated the efficacy of \algoname{DegreeSketch} to scalably and approximately answering queries on massive graphs. 
Although we have focused on estimating neighborhood sizes and counting triangles, \algoname{DegreeSketch}'s utility extends to more general queries that can be phrased as unions and possibly an intersection of adjacency sets. 
In particular, although we have focused on simple undirected graphs in this work, colored graphs are an interesting area of future work.
A simple generalization to the work here presented allows us to estimate interesting queries of the form "how many of $x$'s $t$-neighbors are both red and green?" or "how many of $x$'s $t$-neighbors are not blue?"
\algoname{DegreeSketch}'s demonstrated scalability coupled with its demonstrated performance should prove useful in applications where such queries are prevalent, such as motif counting.

\section{Acknowledgments}
This work was performed under the auspices of the U.S. Department of Energy by Lawrence Livermore National Laboratory under Contract DE-AC52-07NA27344 (LLNL-CONF-757958). 
Experiments were performed at the Livermore Computing Facility.

\newpage
\bibliographystyle{arxiv}
\bibliography{bibliography.bib} 

\newpage
\begin{appendix}

\section{Vertex-Local Variance Bound}
\label{apdx:bound}

The following theorem bounds the estimator error variance of Algorithm~\ref{alg:ds:vertex_local} in terms of the variances of the atomic edge-local estimates using the subadditivity of the standard deviation - i.e. if random variables $a$ and $b$ have finite variance, $\sqrt{\Var \left [ a + b\right ]} \leq \sqrt{\Var \left [ a\right ]} + \sqrt{\Var \left [ b\right ]}$:

\begin{theorem} \label{thm:vertex_local:variance}
	Let $\widetilde{\mathcal{T}}(x)$ be the estimated output of Algorithm~\ref{alg:ds:vertex_local} for $x \in \mathcal{V}$, and let $\widetilde{\mathcal{T}}(xy)$ be the estimated edge triangle count for $xy \in \mathcal{E}$.
	Assume further that for each $xy$, we know a standard deviation bound $\eta_{xy}$ so that
	\begin{equation} \label{thm:vertex_local:variance:bound}
		\frac{\sqrt{\Var \left [ \widetilde{\mathcal{T}}(xy) \right ] }}{\mathcal{T}(xy)} \leq \eta_{xy}.
	\end{equation}
	Furhermore, let $\eta_{*} = \max_{xy \in \mathcal{E}} \eta_{xy}$.
	Then, $\widetilde{\mathcal{T}}(x)$ has at most twice this maximum standard deviation.
	That is,
	\begin{equation*}
		\frac{\sqrt{\Var \left [ \widetilde{\mathcal{T}}(x) \right ] }}{\mathcal{T}(x)} \leq 2\eta_{*}.
	\end{equation*}
\end{theorem}

\begin{proof}
\begin{align*}
	\frac{\sqrt{\Var \left [ \widetilde{\mathcal{T}}(x) \right ] }}{\mathcal{T}(x)}
	&=
	\frac{\sqrt{\Var \left [ \sum\limits_{xy \in \mathcal{E}} \widetilde{\mathcal{T}}(xy) \right ] }}{\mathcal{T}(x)}
	& \\
	&\leq
	\frac{\sum\limits_{xy \in \mathcal{E}} \sqrt{\Var \left [ \widetilde{\mathcal{T}}(xy) \right ] }}{\mathcal{T}(x)}
	& \textnormal{subadditivity} \\
	&\leq
	\frac{\sum\limits_{xy \in \mathcal{E}} \eta_{xy} \mathcal{T}(xy)}{\mathcal{T}(x)}
	& \textnormal{Equation~\eqref{thm:vertex_local:variance:bound}}  \\
	&\leq
	\frac{\eta_* \sum\limits_{xy \in \mathcal{E}} \mathcal{T}(xy)}{\mathcal{T}(x)}
	&  \\
	&=
	2\eta_*
	&  \textnormal{Equation~\eqref{eq:sum_of_edge_tris}}\\
\end{align*}
\end{proof}

Theorem~\ref{thm:vertex_local:variance} shows that if we can bound the error variance of the edge-local triangle count estimates produced using \algoname{DegreeSketch}, we can also bound the error variance of the vertex-local triangle count estimates produced by Algorithm~\ref{alg:ds:vertex_local}. 
Unfortunately, we are unable to provide these bounds a priori, as they depend upon the sizes of all of the the sets and their intersections, which are an unknown function of the graph. 
However, if we are somehow promised that the triangle density of every edge is above a given threshold, we are able to produce analytic guarantees of the estimation error.

\section{Dominations and Small Intersections}
 \label{apdx:intersections}

We have noted that there are limitations to the sketch intersection estimation in Section~\ref{sec:HLL:intersection}.
There appear to be two main sources of large estimation error in practice.
Throughout we will borrow the parlance of Section~\ref{sec:HLL:intersection}.

The first source of error is the phenomenon where $\mathbf{r}^{(A)}_i > \mathbf{r}^{(B)}_i$ for all $i$ where $\mathbf{r}^{(B)}_i > 0$, resulting in  $\mathbf{c}^{(A),<}_k = \mathbf{c}^{(B),>}_k = 0$ for all $k$ and $\mathbf{c}^{=}_k = 0$ for all $k > 0$. 
This generally only occurs when $|A| \gg |B|$ or $B \subseteq A$.
We say that such an $A$ \emph{strictly dominates} $B$. 
In this case, Equation~(70) of \citep{ertl2017new} can be rewritten as the sum of functions depending upon $\lambda_a$ and $\lambda_b + \lambda_x$. 
This means that the optimization relative to $\lambda_a$ does not depend upon $\lambda_x$ or $\lambda_b$.
The optimization relative to $\lambda_b + \lambda_x$ is similarly independent of $\lambda_a$, and thus is tantamount to using the maximum likelihood estimator for $B$ independent of $A$. 
Consequently, $\lambda_x$ could be anything between 0 and $\widetilde{\lambda}_{(B)}$ without affecting the optimum joint likelihood, resulting in arbitrary estimates for the intersection.

\begin{figure}[t]
\centering
\includegraphics[width=0.5\columnwidth]{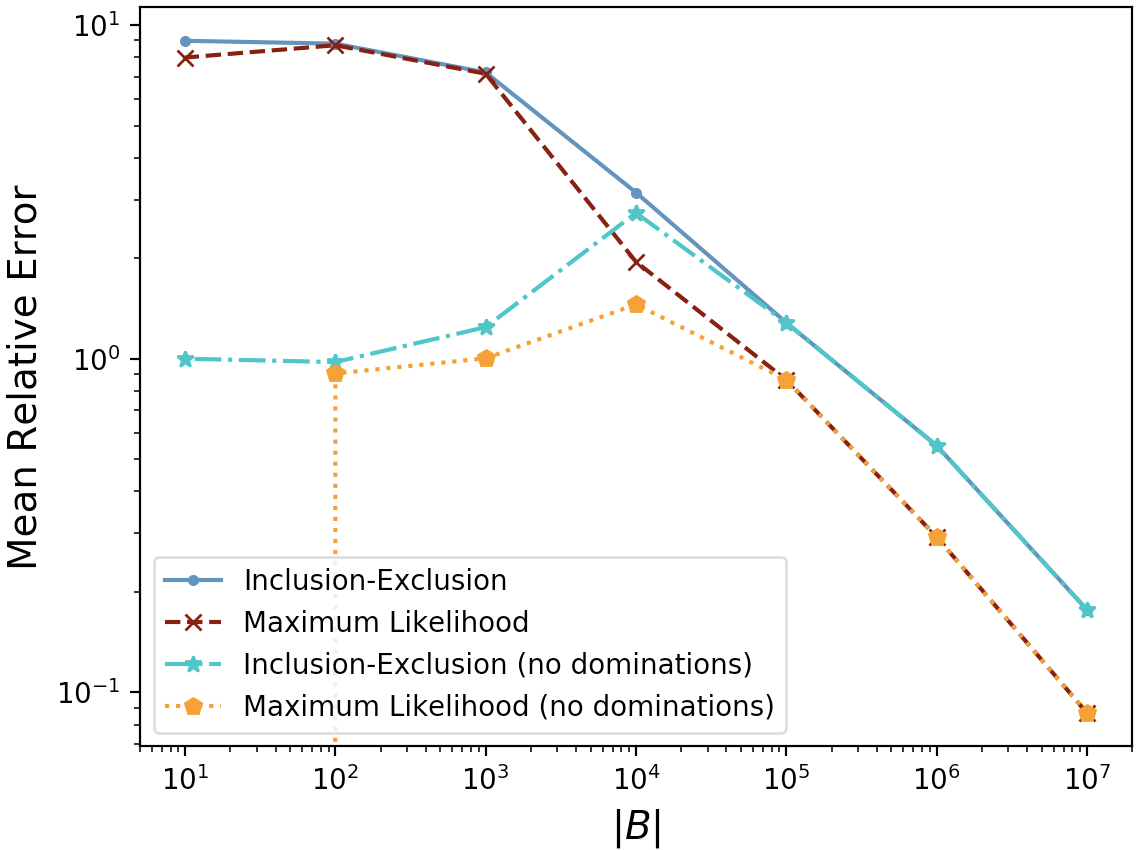}
\caption{Mean relative intersection error as a function of $|B|$, where $|A \cap B| = \frac{|B|}{10}$.
}
\label{fig:domination}
\end{figure}

We also consider the phenomenon where $\mathbf{r}^{(A)}_i \geq \mathbf{r}^{(B)}_i$ for all $i$, resulting in  $\mathbf{c}^{(A),<}_k = \mathbf{c}^{(B),>}_k = 0$ for all $k$. 
We say that such an $A$ \emph{dominates} $B$. 
We are unable to make the same analytic statements about Equation~(70), as the terms dependent upon $\mathbf{c}^{=}$ are not eliminated.
Consequently, the optimum estimate for $\lambda_a$ depends upon $\lambda_b$ and $\lambda_x$. 
If $A$ dominates $B$, the count statistics given by Equation~\ref{eq:cs} are unable to distinguish whether $B$ is subset of $A$. 
Many and large nonzero values for $\mathbf{c}^{=}_k$ for large $k$ will bias the optimization towards larger intersections, whereas the converse is true if $\mathbf{c}^{=}_k$ is nonzero for only a few small values of $k$.
If $|A| \gg |B|$, then the latter might occur whether $|A \cap B|$ is large or small.
Furthermore, note that if $B \subseteq A$, then $A$ will (possibly strictly) dominate $B$. 

If $A$ dominates $B$, then $S^{(A \cup B)} = S^{(A)}$.
Ergo, the inclusion-exclusion estimator returns the estimated value of $B$.
This estimate is dubious, given that we have no evidence that the sets $A$ and $B$ hold any elements in common. 
This is especially true if $|A| \gg |B|$.
Hence, both the na\"ive and maximum likelihood estimators may suffer from bias when a domination event occurs.
Figure~\ref{fig:domination} plots the mean relative error as one of the sets decreases, with a fixed relative intersection size.
As $|B|$ gets smaller, the likelihood of a domination increases. 
At $|B|= 10^4$ dominations occur in $6.6\%$ of cases, at $|B|= 10^3$ dominations occur in $76.9\%$ of cases, at $|B|= 10^2$ dominations occur in $97.5\%$ of cases, and at $|B|= 10$ dominations occur in $99.8\%$ of cases.
In particular in the two cases where $|B|=10$ and $|A \cap B| = 1$ and a domination does not occur, the maximum likelihood estimator returns exactly 1. 
So for a fixed intersection size relative to $|B|$, both the inclusion-exclusion and maximum likelihood estimators return more reasonable estimates when dominations do not occur.
However, there is no known reliable method to avoid them in practice.

Consequently, it might be safest to disregard dominations in practice, as failing to do so is theoretically unsound and likely to produce high and arbitrary error.
However, this poses a problem for many graph applications, as one will frequently have to compare the sketches of high degree vertices with those of comparatively low degree to find their joint triangle count.

\begin{figure}
\centering
\includegraphics[width=0.5\columnwidth]{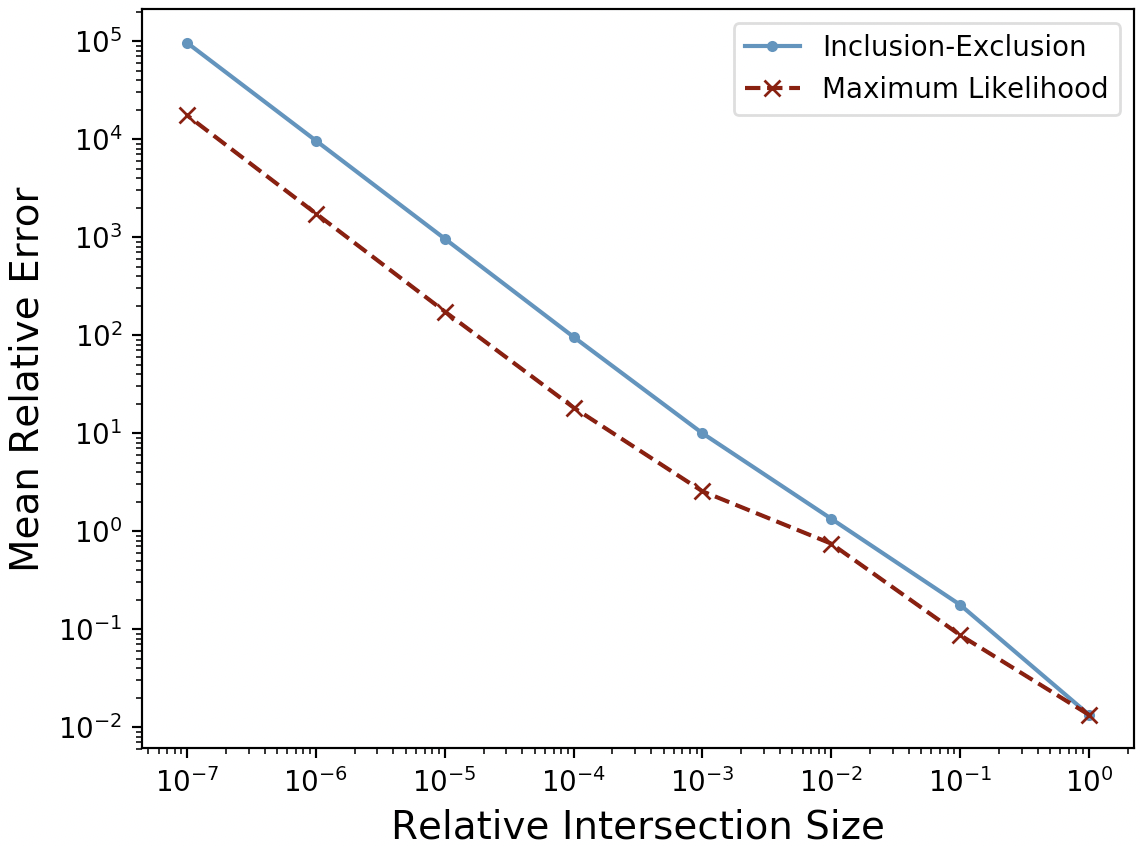}
\caption{\algoname{HLL} inclusion-exclusion and maximum likelihood intersection estimator performance where $|A| = |B| = 10^7$ and $|A \cap B|$ varies from $1$ up to $|B|$.
}
\label{fig:mle}
\end{figure}

We have also noted the problem of small intersections. 
As discussed above, the maximum likelihood intersection estimate is proportional to the number (and size) of the nonzero $\mathbf{c}^{=}_k$ for $k>0$, where larger $k$ biases the estimate toward larger intersections. 
If the ground truth intersection is small relative to $|A|$ and $|B|$, however, Equation~(70) will exhibit high variance.

Figure~\ref{fig:mle} compares the performance of the inclusion-exclusion estimator to the maximum likelihood estimator for a prefix size of 12.
Here the set sizes are kept constant at $10^7$
Note that the mean relative error grows quite large as the relative interection size decreases, although the maximum likelihood estimator consistently outperforms the inclusion-exclusion estimator by roughly an order of magnitude.

\section{Kronecker Graph Construction}
 \label{apdx:kronecker}

Nonstochastic Kronecker graphs \citep{weichsel1962kronecker} have adjacency matrices $C$ that are Kronecker products $C = C_1 \otimes C_2$, where the factors are also adjacency matrices.
This type of synthetic graph is attractive for testing graph analytics at massive scale \citep{leskovec2010kronecker, kepner2018design}, as ground truth solution is often cheaply computable.
For such graphs, global triangle count and triangle counts at edges are computed via Kronecker formulas \citep{sanders2018large}: for a graph with $m$ edges, the worst-case cost of computing global triangle counts is sublinear, $O \left ( m^{\frac{3}{4}} \right )$, whereas the cost of computing the full set of edge-local counts is $O \left ( m \right )$.

Here, we build $C = C_1 \otimes C_2$ from identical factors, $C_1=C_2$, that come from a small set of  graphs with $m$ up to $10^5$ from the University of Florida sparse matrix collection ({\tt polbooks}, {\tt celegans}, {\tt geom}, {\tt yeast} \citep{davis2011university}). 
All graphs were forced to be undirected, unweighted, and without self loops.  
We compute the number of triangles at each edge for $C_1$ and use the Kronecker formula in \citep{sanders2018large} to get the respective quantities for $C$.
Summing over the edges and dividing by 3 gives the global triangle count for $C$.

\end{appendix}

\end{document}